\documentclass{article}
\usepackage[nolists,nomarkers,tablesfirst]{endfloat}
\usepackage{amsfonts}
\usepackage{amsmath}
\usepackage{amsthm} 
\usepackage{mathtools}
\usepackage{graphicx} 
\usepackage{caption}
\usepackage{natbib}
\usepackage{mathrsfs}
\usepackage{bm} 
\usepackage[margin=1in]{geometry}
\usepackage[inline]{enumitem}
\usepackage{hyperref}
\hypersetup{
    colorlinks=true,
    citecolor=blue,
}
\renewcommand{\Im}{\operatorname{Im}}
\renewcommand{\Re}{\operatorname{Re}}
\newcommand{\e}{{\rm e}}
\newcommand{\im}{{\rm i}}
\newcommand{\E}{{\mathbb E}}
\newcommand{\Pa}{{\mathbb P}}
\newcommand{\Q}{{\mathbb Q}}

\newcommand{\R}{{\mathbb R}}

\newcommand{\N}{{\mathbb N}}

\newcommand{\Ecal}{{\mathcal E}}
\newcommand{\Fcal}{{\mathcal F}}
\newcommand{\Gcal}{{\mathcal G}}

\newcommand{\Lcal}{{\mathcal L}}

\newcommand{\Ncal}{{\mathcal N}}

\newcommand{\Ind}{{ 1}}

\DeclareMathOperator{\sign}{sign}

\DeclareMathOperator{\diag}{diag}
\DeclareMathOperator*{\argmin}{arg\,min}
\newtheorem{proposition}{Proposition}[section]
\newtheorem{lemma}[proposition]{Lemma}
\newtheorem{theorem}[proposition]{Theorem}

\newtheorem{remark}[proposition]{Remark}
\newtheorem{exampleemph}[proposition]{Example}   

\newenvironment{example}{\begin{exampleemph}\begin{upshape}}{\end{upshape}\end{exampleemph}} 
\newtheorem{foo}[proposition]{Remarks}

\usepackage{array}
\newcolumntype{L}{>{$}l<{$}}

\date{May 20, 2019}
\title{Option Pricing with Orthogonal Polynomial Expansions\footnote{We thank the participants at the 2017 CIB workshop on Dynamical Models in Finance in Lausanne and at the Bachelier World Congress 2018 in Dublin, as well as Christian Gollier, Vadim Linetsky, Fabio Nobile, Uffe H{\o}gsbro Thygesen, Sander Willems, and two anonymous referees for comments.
The research leading to these results has received funding from the European Research Council under the European Union's Seventh Framework Programme (FP/2007-2013) / ERC Grant Agreement n.~307465-POLYTE.}}
\author{Damien Ackerer\footnote{Swissquote Bank, 1196 Gland, Switzerland. {\it Email: }damien.ackerer@swissquote.ch} \and Damir Filipovi\'c\footnote{EPFL and Swiss Finance Institute, 1015 Lausanne, Switzerland. {\it Email: }damir.filipovic@epfl.ch}}

\begin{document}

\maketitle

\begin{center} 
\large forthcoming in \textit{Mathematical Finance} 
\end{center} 

\bigskip 

\begin{abstract}
We derive analytic series representations for European option prices in polynomial stochastic volatility models.
This includes the Jacobi, Heston, Stein--Stein, and Hull--White models, for which we provide numerical case studies.
We find that our polynomial option price series expansion performs as efficiently and accurately as the Fourier transform based method in the nested affine cases.
We also derive and numerically validate series representations for option Greeks.
We depict an extension of our approach to exotic options whose payoffs depend on a finite number of prices.
\end{abstract}




\section{Introduction}

We study the class of stochastic volatility models for asset price log returns $dX$ in which the squared volatility, $d\langle X\rangle/dt$, is given by a polynomial in an underlying univariate polynomial diffusion $Y$. This class nests many models of interest having bounded and unbounded volatility support, including the Jacobi, Heston, Stein--Stein, and Hull--White models. Such polynomial models are highly tractable because the conditional moments of the marginal distributions of the log price $X_T$ are given in analytic form in terms of a matrix exponential. We use this to derive the prices of options with discounted payoff $f(X_T)$ in terms of analytic orthogonal polynomial expansions. Our approach consists of two steps. Observing that $X_T$ has a density $g$ given by a Gaussian mixture with an infinite number of components, we first approximate $g$ by an auxiliary, finite Gaussian mixture density $w$. We provide explicit conditions under which the likelihood ratio function $\ell=g/w$, as well as the discounted payoff function $f$, lie in the weighted space $L^2_w$ of square-integrable functions with respect to $w$. The corresponding option price then admits a series representation of the form $\pi_f=\langle f,\ell\rangle_w=\sum_{n\ge0} f_n\ell_n$ where the Fourier coefficients $f_n$ and $\ell_n$ of $f$ and $\ell$ are with respect to an orthonormal basis of polynomials $\{H_n,\ n\ge 0\}$ of $L^2_w$. Thanks to the polynomial property of the model, the likelihood coefficients $\ell_n=\E[H_n(X_T)]$ are given in analytic form, in terms of a matrix exponential. For important examples, including European call options, the payoff coefficients $f_n$ are also given in analytic form. The option price $\pi_f=\sum_{n\ge0} f_n\ell_n$ is approximated by truncating the series at a finite order $N$. This option price approximation is accurate for a small truncation order $N$ if the true density $g$ is statistically close to the auxiliary density $w$. Even if $\ell$ does not belong to $L^2_w$, so that the sequence of option price approximations $\sum_{n=0}^N f_n\ell_n$ does not converge as $N\to\infty$, we show that a finite order $N$ can provide an accurate approximation of the true price $\pi_f$. We also provide simple and efficient algorithms to construct the auxiliary mixture density $w$ and the corresponding orthonormal polynomials. We derive recursive systems of equations for the payoff coefficients $f_n$ in case of a European call option and for auxiliary Gaussian and Gamma mixture densities.

We find that the option price sensitivities $\partial_\theta\pi_f$ with respect to specific model parameters $\theta$, also known as option Greeks, admit a series representation in terms of the derivatives of the likelihood coefficients, $\partial_\theta \ell_n$. For $\theta$ being the initial asset price this gives the option Delta and carries over to higher order derivatives, such as the option Gamma. We find that the Delta and Gamma are easy to calculate and given in analytic form. In general, the computation of the derivatives $\partial_\theta \ell_n$ boils down to the computation of the Fr\'echet derivative of a matrix exponential. But this problem is numerically well addressed by complex-step differentiation, which in turn renders practical option Greek calculation and gradient based optimization for model calibration.

We also discuss an extension of our approach to construct option price series representations for options whose payoffs depend on multiple prices. We first de-correlate the respective log returns via a linear transformation. We let then the corresponding multivariate auxiliary density $w$ be a tensor product of univariate densities. The multivariate orthonormal polynomial basis is also given by a tensor product of the univariate bases, and the likelihood coefficients $\ell_n$ are again given in analytic form. The payoff coefficients $f_n$ may however not be in analytic form in higher dimensions, and we provide an effective weighted least squares algorithm to approximate them.

We validate our method with several numerical applications. In the Jacobi model with a single auxiliary Gaussian density $w$, the accuracy of the European call option price approximation for a fixed $N$ decreases rapidly as the upper bound of the volatility support increases.
We therefore let the auxiliary density $w$ be a mixture of two Gaussian distributions whose first two moments are matching the log price density $g$.
We show that the option price series converges significantly faster using this Gaussian mixture density in comparison to a single Gaussian density.
We also approximate European call option prices in the Stein--Stein model using Gaussian mixtures with varying number of components and for various parameter choices, and benchmark them to the option prices computed using Fourier transform techniques. The latter is possible as the Stein--Stein model can be embedded in an affine framework. We also price European call options in the Hull--White model, which is shown to induce large log price kurtosis values. Despite the fact that the likelihood ratio function $\ell$ in these models does not belong to $L^2_w$, where $w$ is a Gaussian mixture density, we find that our approach still produces accurate option price approximations. We also price European call options on the realized variance in the GARCH diffusion model.
We compare our Delta and Gamma approximations with the values obtained by Fourier transform in the affine Heston model and find them to be quantitatively equivalent.
We compute option price sensitivities with respect to $\theta$ and $\rho$ in the Heston model for many parameter values and via complex-step differentiation in order to illustrate the benefits for model analysis.
We additionally show that one can perform model calibration faster with our approach than with the Fourier transform approach on data samples with more than 30 options.
We benchmark the approximation of the payoff coefficients with the weighted least squares algorithm against the recursive formulas for the call option and for the forward start call option when the auxiliary density is a Gaussian with diagonal covariance matrix, we find that the option price approximations are almost equal.

Stochastic volatility models for asset returns are necessary to reproduce realistically the dynamical behavior of asset prices.
Such models were first introduced by~\cite{hull1987pricing}, \cite{scott1987option}, \cite{wiggins1987option} to overcome the assumption of constant volatility in the seminal works by \cite{black1973pricing}, and~\cite{merton1973theory}.
However, in the context of derivatives, a stochastic volatility model is practical only if tractable formulas are available to approximate option prices.
The \cite{heston1993closed} model therefore became a benchmark because it has a quasi-analytical formula to price European options.
The Heston model belongs to the large class of affine models whose development led to many new tractable models, see \cite{duffie2003affine}.
Indeed, the characteristic function of the log price can be computed numerically in affine models which enables the use of standard Fourier transform techniques to compute European option prices, see for details~\cite{carr1999option,fang2009novel,lipton2008stochastic}.
Option pricing via orthogonal polynomials in the Heston model with jumps has also been proposed in~\cite{carr2008hedging}.
Our paper adds to this literature by extending the range of tractable models beyond the affine class. The literature on option pricing in non-affine polynomial models contain \cite{ackerer2018jacobi} for the Jacobi model or \cite{filipovic2016quadratic} for variance swaps. Our results improve the option pricing performance in~\cite{ackerer2018jacobi}.
Notable progress has recently been made on the pricing of early-exercise and path-dependent options with stochastic volatility using recursive marginal quantization as studied in \cite{pages2015recursive,mcwalter2017recursive,callegaro2017pricing}.
In the context of affine and polynomial models, this approach has been shown to perform well when combined with Fourier transform techniques as in~\cite{callegaro2017quantization}, or polynomial expansion techniques as in~\cite{callegaro2017quantizationb}, whose results could be further improved with the new expansions presented in our paper.
The calculation of Greeks for stochastic volatility models is a difficulty task often adressed by Monte Carlo simulations, see for examples \cite{broadie2004exact} for a discussion of different simulation based estimators in the Heston model, and \cite{chan2015first} for more recent advances using algorithmic differentiation.

Our approach builds on Hilbert space methods and thus relates to a large literature on eigenfunction expansion methods where the basis functions are exact eigenfunctions of the pricing semigroup, see \cite{lewis1998applications,davydov2003pricing,linetsky2004lookback,linetsky2004spectral,li2015discretely,lorig2014pricing}.
Although this approach successfully addressed complex pricing problems such as Asian and lookback options, the exact eigenfucntions of Markov semigroups are known in analytic form only for a few exceptional cases.
In contrast, our approach offers a hands-on method that works for a large class of models and gives analytic expressions, in terms of series, for option prices based on matrix exponentials or, equivalently, on solutions to linear ODEs.
Hilbert space approximation methods have also been extensively studied in the PDE literature where a PDE problem is typically reformulated as a variational problem and a solution is sought as a projection on a finite set of basis functions, often orthogonal polynomials.
We refer the interested reader to~\cite{feng2007variational,cohen2011analytic,back2011stochastic,beck2012optimal}.

The remainder of the paper is as follows.
Section~\ref{sec:expan} presents the density expansion and option price series representation with an auxiliary mixture density.
Section~\ref{sec:polSV} introduces the polynomial stochastic volatility models, derives the option Greek series representation, and describes some Gaussian mixture constructions for the auxiliary density.
Section~\ref{sec:ext} discusses an extension to multivariate payoff functions and the payoff coefficients approximation with weighted least squares methods.
Section~\ref{sec:MIXapp} contains the numerical applications.
Section~\ref{sec:ccl} concludes.
The proofs are collected in Appendix~\ref{sec:MixProofs}.
Alternative moment based methods to construct orthonormal polynomial basis can be found in Appendix~\ref{sec:ONBmts}.
The Fourier transform formulas for the option price, Delta, and Gamma in the Heston model are reported in Appendix~\ref{sec:Heston}.

\section{Polynomial Price Series Expansions} \label{sec:expan}

We recap the density expansion approach described in~\cite{filipovic2013density} along with the option price series representation further developed in~\cite{ackerer2018jacobi}.
We show how the elements of these option price series can be efficiently computed when the auxiliary density is a mixture density.
We then give some examples of component densities which are convenient to work with.

\subsection{European Option Price Series Representation}

Fix a maturity $T>0$ and assume that the distribution of the log price $X_T$ of an asset has a density $g(x)$.
Our goal is to compute the option price
\begin{equation} \label{eq:pricing}
\pi_f =  \int_\R f(x)g(x)dx
\end{equation}
for some discounted payoff function $f(x)$.
Let $w(x)$ be an auxiliary density that dominates $g(x)$ with likelihood ratio $\ell(x)$ given by
\[
g(x) = \ell(x) w(x).
\]
We define the weighted Lebesgue space,
\[
L_w^2 = \left\{ f(x) \mid \lVert f \rVert^2_w = \int_\R f(x)^2 w(x)dx < \infty
 \right\},
\]
which is a Hilbert space with the scalar product
\[
\langle f,g\rangle_w = \int_\R f(x) g(x) w(x)dx.
\]
Assume that the polynomials are dense in $L^2_w$ and that the functions $\ell(x)$ and $f(x)$ are in $L_w^2$.
Then, the option price has the following series representation
\begin{equation} \label{eq:pif}
\pi_f = \sum_{n\ge 0} f_n\ell_n
\end{equation}
with \textit{likelihood coefficients}
\begin{equation} \label{eq:elln}
\ell_n=\langle \ell,H_n\rangle_w=\int_\R H_n(x) g(x)dx
\end{equation}
and \textit{payoff coefficients}
\begin{equation} \label{eq:fn}
f_n = \langle f, H_n\rangle_w = \int_\R f(x) H_n(x) w(x)dx,
\end{equation}
and where $H_0=1, H_1, H_2,\dots $, denotes an orthonormal polynomial basis of $L^2_w$ such that $\deg H_n(x)= n$ and $\langle H_m,H_n\rangle_w=1$ if $m=n$ and zero otherwise.
In practice we approximate the option price by truncating the option price series in~\eqref{eq:pif} as follows
\begin{equation} \label{eq:piN}
\pi_f^{(N)} = \sum_{n=0}^{N} f_n \ell_n
\end{equation}
for some positive integer $N$.
The accuracy of this approximation crucially depends on the statistical distance between the true density and the auxiliary density. Indeed, as $\ell_0=1$, the $L^2$-divergence of $g$ from $w$ equals $\| \ell-1\|_w^2 = \sum_{n\ge 1}\ell_n^2$. For example if $w(x)=g(x)$ then $\pi_f=f_0\ell_0=f_0$, as $\ell_n=0$ for all $n\ge 1$ in this case.
The efficiency of this approach also depends on how fast the operations required to compute the coefficients in~\eqref{eq:elln}~and~\eqref{eq:fn} can be performed.

\subsection{Auxiliary Mixture Density}

We now assume that the auxiliary density $w(x)$ is a mixture density of the form
\[
w(x) := \sum_{k=1}^K c_k v_k(x)
\]
for some mixture weights $c_k>0$ satisfying $\sum_{k=1}^K c_k=1$, and some mixture components $v_k(x)$ which are also probability densities.
To each density $v_k(x)$ is associated an orthonormal polynomial basis $H_n^k(x)$, $n\ge0$.
Let $a_n^k$ and $b_n^k$ denote the coefficients that define the three term recurrence relation for $H_n^k(x)$, which always holds for univariate orthonormal polynomial basis, see e.g.\ \cite[Section~1.3]{gau_04},
\begin{equation} \label{eq:ONBkrec}
x H_{n}^k(x) = b_{n+1}^k H_{n+1}^k(x) + a_n^k H_{n}^k(x) + b_n^k H_{n-1}^k(x)
\end{equation}
for all $n\ge 0$ with $H_{-1}^k=0$ and $H_{0}^k=1$.
We define the following tridiagonal matrices
\[
J^k_N = \begin{pmatrix}
a_0^k & b_{1}^k & & &  \\
b_{1}^k & a_1^k & b_2^k & & \\
 & \ddots & \ddots & \ddots & \\
 & & b_{N-1}^k & a_{N-1}^k & b_N^k\\
 & & & b_N^k & a_N^k
\end{pmatrix}, \quad \text{for $k=1,\dots,K$.}
\]
The recurrence relation~\eqref{eq:ONBkrec} used to construct the orthonormal polynomial basis is explicitly known for many densities taking values both on compact and unbounded supports.
When $v_k(x)$ is a Gaussian density with mean $\mu_k$ and variance $\sigma_k^2$, the coefficients in~\eqref{eq:ONBkrec} are given by $a^k_n = \mu_k$ and  $b_n^k = \sqrt{n} \sigma_k$.
We refer to \cite[Chapter 1]{schoutens2012stochastic} for an overview of orthonormal polynomial basis families.

\begin{remark}
The coefficients $a_n^k$ and $b_n^k$ can be inferred from the basis $H^k_n(x)$.
Denote $\alpha^k_{n,i}$ the coefficient in front of the monomial $x^i$ of the polynomial $H^k_n(x)$.
Then by inspection of the recurrence relation~\eqref{eq:ONBkrec} we have
\[
b^k_{n+1} = \frac{\alpha^k_{n,n}}{\alpha^k_{n+1,n+1}} \quad \text{and} \quad a^k_n = \frac{\alpha^k_{n,n-1} - b^k_{n+1}\alpha^k_{n+1,n}}{\alpha^k_{n,n}}.
\]
\end{remark}

The following proposition gives an algorithm to compute the coefficients $a_n$ and $b_n$ in the recursion of the orthonormal polynomial basis $H_n(x)$ associated with the mixture density $w(x)$,
\[ x H_{n}(x) = b_{n+1} H_{n+1}(x) + a_n H_{n}(x) + b_n H_{n-1}(x)\]
for all $n\ge 0$ with $H_{-1}=0$ and $H_{0}=1$.

\begin{proposition} \label{prop:ONBmix}
The recurrence relation coefficients of the mixture density $w(x)$ are given by $b_n = \sqrt{\psi_n/\psi_{n-1}}$ for $n=1,\dots,N$ and  $a_n=\phi_n/ \psi_n$ for $n=0,\dots,N-1$ where
\begin{align*}
\psi_n &= \sum_{k=1}^K c_k \,(z_n^k)^\top \, z^k_n\\
\phi_n &= \sum_{k=1}^K c_k \, (z_n^k)^\top \, J_N^k \, z^k_n\\
z_{n+1}^k &= (J_N^k - a_n I)z^k_n - (b_n)^2 z^k_{n-1} \quad \text{for all $k=1,\dots,K$}
\end{align*}
with $z_{-1}^k=0$, and $z_0^k=\e_1$ is the vector whose first coordinate is equal to one and zero otherwise.
\end{proposition}
This algorithm is fast and performs well numerically.
For example, with $N,K\sim10^2$, it takes few milliseconds on a modern CPU to construct the orthonormal basis.
There are other moments based approaches to construct the orthonormal basis, such as the Gram-Schmidt and the Mysovskikh algorithms, which are described in Appendix~\ref{sec:ONBmts}.
However these methods are typically subject to numerical problems and may be slow even for a relatively small order $N$.

The following proposition shows that the payoff coefficients with respect to the mixture density can be efficiently computed when the corresponding coefficients are known for the mixture components.

\begin{proposition}\label{prop:fnMIX}
The payoff coefficients are equal to
\begin{equation}\label{eq:fnMIX}
f_N = \langle f, H_N \rangle_{w}= \sum_{k=1}^K \sum_{n=0}^N \, c_k \, q_{N,n}^k \,f_n^k \quad \text{with }f_n^k = \langle f, H_n^k \rangle_{v_k}
\end{equation}
and where  $q^k_N\in\R^N$ is the vector representation of $H_N(x)$ in the basis $H_n^k(x)$
\begin{equation}\label{eq:Hnbasisk}
H_N(x) = \sum_{n=0}^N q_{N,n}^k \, H_n^k(x), \quad k=1,\dots,K.
\end{equation}
\end{proposition}
The usefulness of Proposition~\ref{prop:fnMIX} lies on the premise that the coefficients $f_n^k$ can be easily computed.
In a situation where they are numerically costly to compute then one may directly integrate the payoff function with respect to the following density approximation
\begin{equation}\label{eq:densapp}
w^{(N)}(x) =  \sum_{n=0}^N \ell_n H_n(x) \sum_{k=1}^K c_k v_k(x).
\end{equation}

\begin{remark}
The vectors $q^{k}_N$ can be efficiently computed.
Let $\mathbf{H}^k_N\in\R^{(N+1)\times(N+1)}$ denote the matrix whose $(i,j)$-th element is given by the coefficient in front of the monomial $x^{j-1}$ in the $(i-1)$-th polynomial of the basis $H_n^k(x)$.
Define similarly the matrix $\mathbf{H}_N$ with respect to the polynomial basis $H_n(x)$.
The matrices $\mathbf{H}^k_N$ for $k=1,\dots,K$, and $\mathbf{H}_N$ are upper triangular.
We are interested in the upper triangular matrix $\mathbf{Q}^k_N\in\R^{(N+1)\times(N+1)}$ for $k=1,\dots,K$ whose $(i,j)$-th element is equal to $q^k_{j,i}$.
It is also given by
$
\mathbf{H}_N = \mathbf{H}_N^k \, \mathbf{Q}_N^k
$
which forms a triangular system of equations and can thus be solved efficiently.
\end{remark}

\subsection{Examples of Mixture Components}

In practice, to ensure efficient option prices approximations we need to select mixture components $v_k(x)$ whose orthonormal polynomial bases $H_n^k(x)$ and payoff coefficients $f_n^k$ can be efficiently computed.
The Gaussian and uniform distributions have been successfully used in~\cite{ackerer2018jacobi} and in~\cite{ackerer2016linear} respectively.
The logistic distribution whose tail decreases at an exponential-linear rate is used in \cite{heston2016spanning}. But the corresponding payoff coefficients are given by complicated expressions involving special functions.
The bilateral Gamma distribution of~\cite{kuchler2008bilateral} was used in~\cite{filipovic2013density} to accurately approximate option prices by numerical integration of the discounted payoff function with respect to the density approximation.
However an explicit recursion formula to construct the orthonormal polynomial basis was not provided.

We study hereinbelow the Gamma distribution whose tail decays at a polynomial-exponential-linear rate and for which simple recursive expressions can be derived for the payoff coefficients.
By mixing two Gamma distributions one can obtain a distribution on the entire real line.
The Gamma density on the half-line $(\xi,+\infty)$ for some $\xi\in\R$ is defined by
\begin{equation}\label{eq:Gammadens}
v_k(x) = \Ind_{\{x>\xi\}} \, \frac{\beta^\alpha}{\Gamma(\alpha)}(x-\xi)^{\alpha-1}\e^{-\beta(x-\xi)}
\end{equation}
for some shape parameter $\alpha\ge 1$, rate parameter $\beta>0$, and where $\Gamma(\alpha)$ is the upper incomplete Gamma function defined by
\[
\Gamma(\alpha) = \Gamma(\alpha,0) \quad \text{with} \quad \Gamma(\alpha,z) =\int_z^\infty x^{\alpha-1}\e^{-x}dx,
\]
so that $\Gamma(n)=(n-1)!$ for any positive integer $n$.
The Gamma density $v_k(x)$ admits an orthonormal polynomial basis $H^k_n(x)$ given by
\[
H^k_n(x) = \sqrt{\frac{n!}{\Gamma(\alpha+n)}} \, \Lcal^{\alpha-1}_n(\beta(x-\xi))
\]
where $\Lcal^{\alpha-1}_n(x)$ denotes the $n$-th order generalized Laguerre polynomial with parameter $\alpha-1$ defined by
\[
\Lcal^{\alpha-1}_n(x) = \frac{x^{-\alpha+1}\e^x}{n!} \frac{\partial^n}{\partial x^n} (\e^{-x}x^{\alpha-1+n}).
\]
The generalized Laguerre polynomials are recursively given by
\begin{align*}
\Lcal^{\alpha-1}_0 (x) &= 1\\
\Lcal^{\alpha-1}_1 (x) &= \alpha + x\\
\Lcal^{\alpha-1}_{n+1} (x) &= \frac{2n+\alpha-x}{n+1} \, \Lcal^{\alpha-1}_n(x) - \frac{n+\alpha-1}{n+1} \, \Lcal^{\alpha-1}_{n-1}(x) \quad \text{for all $n\ge 1$.}
\end{align*}

The following proposition shows that the payoff coefficients $f^k_n$ of two specific discounted payoff functions can be recursively computed for the Gamma distribution $v_k(x)$.
\begin{proposition}\label{PROP:lag:coefs}
\begin{enumerate}
\item\label{PROP:lag:coefs1} Consider the discounted payoff function of a call option with log strike ${\bar k}$,
\[
f(x) = \e^{-rT} \left( \e^x - \e^{\bar k} \right)^+.
\]
Its payoff coefficients $f_n$ are given by
\begin{align*}
f_n &= \e^{-rT} \sqrt{\frac{{n!}}{\Gamma(\alpha+n) }}\frac{1}{\Gamma(\alpha)} \left(\e^{\xi} I_n^{\alpha-1}\left(\mu ; \beta^{-1} \right) + \e^{\bar k}I_n^{\alpha-1}\left( 0 ; \beta^{-1} \right) \right) ,\quad n\ge0
\end{align*}
with $\mu=\max(0,\beta({\bar k}-\xi))$ and where the functions $I^{\alpha-1}_n(\mu;\nu)$ are recursively defined by
\begin{equation}\label{eq:Inrecur}
\begin{split}
I^{\alpha-1}_0(\mu;\nu) &= (1-\nu)^{-\alpha} \, \Gamma(\alpha) \,  \Gamma(\alpha,\mu(1-\nu)\\
I^{\alpha-1}_1(\mu;\nu) &= \alpha I_0^{\alpha-1}(\mu;\nu) + I_0^{\alpha}(\mu;\nu)\\
I^{\alpha-1}_n(\mu;\nu)  &= \left(2+\frac{\alpha - 2}{n} \right)I_{n-1}^{\alpha-1}(\mu;\nu) - \left(1 + \frac{\alpha - 2}{n} \right)I_{n-2}^{\alpha-1}(\mu;\nu)  \\
& \quad - \frac{1}{n}\left(I_{n-1}^{\alpha}(\mu;\nu) -I_{n-2}^{\alpha}(\mu;\nu) \right),\quad n\ge 2.
\end{split}.
\end{equation}

\item\label{PROP:lag:coefs2} Consider the discounted payoff function with strike ${\bar k}$,
\[
f(x) = \e^{-rT} \left( x - {\bar k} \right)^+.
\]
Its payoff coefficients $f_n$ are given by
\begin{align*}
f_n = \e^{-rT}  \sum_{j=0}^{n+1} \; p_{n,1+j} \, \frac{\Gamma(\alpha+j,\beta({\bar k}-\xi))}{\Gamma(\alpha) \beta^{j}}, \quad n \ge 0
\end{align*}
where $p_{n}\in\R^{n+2}$ is the vector representation of $x H_n^k(x+\xi) +(\xi-{\bar k})H^k_n(x+\xi)$ in the monomial basis.
\end{enumerate}
\end{proposition}
Note that the calculation of the term $I^{\alpha-1}_{n}(\mu;\nu)$ in part~{\ref{PROP:lag:coefs1}} requires calculations of the terms $I^{\alpha}_{n-1}(\mu;\nu)$ and therefore the dimension of the recursive system grows at the rate $n^2$. Part~{\ref{PROP:lag:coefs2}} will be used to price realized variance derivatives, see example in Section~\ref{sec:num_pricing}.

The Gamma distribution on the half-line $(-\infty, \xi)$ for some $\xi\in\R$ together with its polynomial basis and Fourier coefficients can be similarly derived.
A Gamma mixture on $\R$ can thus be constructed when $K\ge2$ by letting, for example, the first component support be $(-\infty,\xi]$ and the second $[\xi,\infty)$.
In particular, the coefficients in~\eqref{eq:ONBkrec} when $v_k(x)$ is a Gamma density are given by
\begin{equation}\label{eq:gammix_coefs}
a^k_n = \frac{\delta 2n + \delta \alpha + \beta \xi}{\beta},  \quad n\ge 0,
\quad
\text{ and }
\quad
b^k_n = -\frac{\delta \sqrt{(n + \alpha -1) n}}{\beta},\quad n\ge 1,
\end{equation}
where $\delta=1$ if the domain of the Gamma density is $(\xi,\infty)$ and $\delta=-1$ if its domain is $(-\infty, \xi)$.
The derivation of~\eqref{eq:gammix_coefs} is provided in Appendix~\ref{sec:MixProofs}.

\section{Polynomial Stochastic Volatility Models} \label{sec:polSV}

We present a class of stochastic volatility models for which the log price density $g(x)$ is given by a Gaussian mixture with an infinite number of components.
We show that the option price sensitivities, the Greeks, also have series representations whose terms can be computed.
We then describe a simple method to approximate the log price density by an auxiliary Gaussian mixture $w(x)$ with a finite number of components.
The section terminates by studying the option price approximation error when the likelihood ratio $\ell(x)=g(x)/w(x)$ does not belong to the weighted space $L_w^2$.

\subsection{Definition and Basic Properties}

We fix a stochastic basis $(\Omega,\Fcal,\Fcal_t,\Q)$ where $\Q$ is a risk-neutral measure, carrying a multivariate Brownian motion $W_{t}$. We first recall the notion and basic properties of a polynomial diffusion, see \cite{filipovic2017polynomial}. For $d,n\in\N$ and a state space $E\subseteq\R^d$, denote by ${\rm Pol}_n(E)$ the linear space of polynomials on $E$ of degree $n$ or less. Consider an $E$-valued {\emph{polynomial}} diffusion
\[ dZ_t = b(Z_t)\,dt +\sigma(Z_t)\,dW_t \]
in the sense that its generator $\Gcal = b^\top\nabla + (1/2) {\rm Tr}(\sigma\sigma^\top\nabla^2)$ maps ${\rm Pol}_n(E)$ to ${\rm Pol}_n(E)$ for all $n\in\N$. It can easily be verified that this polynomial property holds if and only if $b\in{\rm Pol}_1(E)$ and $a\in{\rm Pol}_2(E)$ componentwise, see \cite[Lemma~2.2]{filipovic2017polynomial}. Fix a degree $n$ and a basis $q_0,\dots,q_M$ of ${\rm Pol}_n(E)$, where $M+1=\dim{\rm Pol}_n(E)$, and write $Q(z)=(q_0(z),\dots,q_M(z))^\top$. Then $\Gcal$ restricted to ${\rm Pol}_n(E)$ has a matrix representation $G$, so that $\Gcal p(z)=Q(z)^\top G\vec p$ for any polynomial $p\in{\rm Pol}_n(E)$ and where $\vec p$ denotes its vector representation. A very useful consequence is that the conditional moments of $Z_T$ are given in analytic form by the moment formula
\begin{equation}\label{momentfor}
  \E[p(Z_T)\mid\Fcal_t] = Q(Z_t)^\top\e^{G(T-t)}\vec p ,
\end{equation}
see \cite[Theorem 2.4]{filipovic2017polynomial}.

Now let $Y_t$ be a univariate polynomial diffusion of the form
\begin{equation} \label{eq:sdeY}
dY_t = \kappa(\theta - Y_t)\,dt + \sigma(Y_t) \,dW_{1t} \quad \text{with} \quad \sigma(y)^2 = \alpha + a y + A y^2
\end{equation}
for some real parameters $\kappa,\theta,\alpha,a,A$.
We then define the dynamics of the log price $X_t$ as follows
\begin{equation*}\label{eq:sdeX}
dX_t = \left(r-\delta \right)dt
- \frac{1}{2} d\langle X \rangle_t  + \Sigma_1(Y_t)\, dW_{1t} + \Sigma_2(Y_t)\, dW_{2t}
\end{equation*}
for the interest rate $r$, the dividend rate $\delta\ge0$, and such that
$
\Sigma_1^2 + \Sigma_2^2 \in {\rm Pol}_m(\R)$ and $ \Sigma_1 \sigma \in {\rm Pol}_{m+1}(\R)$ for some $m\in\N$.
The initial values $(X_0,Y_0) \in \R^2$ are deterministic. It can be verified that the $\R^{m+1}$-valued diffusion $Z_t=(X_t,Y_t,Y_t^2,\dots,Y_t^{m})$ has the polynomial property on $E=\{z=(x,y,y^2,\dots,y^m)\mid  (x,y)\in\R^2\}$, see~\cite[Theorem~5.5]{filipovic2017polynomial}.
Fix the order $N$, the moment formula~\eqref{momentfor} implies that the payoff coefficients in~\eqref{eq:elln} can be computed explicitly as follows
\begin{equation} \label{eq:fnPP}
\ell_n = \E[H_n(X_T)] = Q(Z_0)^\top\e^{G\, T} \vec{H_n}, \quad n=0,\dots,N
\end{equation}
for some chosen polynomials $Q(z)=(q_0(z),\dots,q_M(z))^\top$ that form a basis of ${\rm Pol}_N(E)$. As above, $G$ is the matrix representation of the generator of $Z_t$ and $\vec{H_n}$ is the vector representation of the polynomial $H_n(x)$ in this basis.
Some classical stochastic volatility models are nested in this setup, such as the Heston, Jacobi, Stein--Stein, and Hull--White models.

The spot variance of the log return $dX_t$ is $V_t =d\langle X\rangle_t/dt=\Sigma_1(Y_t)^2 + \Sigma_2(Y_t)^2$.
The leverage effect refers to the generally negative correlation between $dV_t$ and $dX_t$ and is given by
\begin{equation} \label{eq:VXcor}
\begin{aligned}
{\rm lev}(X_t) &= \frac{d\langle V,X\rangle_t}{\sqrt{d\langle V \rangle_t}\;\sqrt{d\langle X\rangle_t}} \\
&= \frac{\Sigma_1(Y_t)}{\sqrt{\Sigma_1(Y_t)^2 + \Sigma_2(Y_t)^2}} \sign\left[\left(\Sigma'_1(Y_t)\Sigma_1(Y_t) + \Sigma'_2(Y_t)\Sigma_2(Y_t)\right)\sigma(Y_t)\right].
\end{aligned}
\end{equation}
The volatility of the volatility of the log return is \begin{equation} \label{eq:volvol}
\begin{aligned}
{\rm volvol}(X_t) &= \sqrt{\frac{d\langle \sqrt{V}\rangle_t}{dt}} = \frac{1}{2 \sqrt{V_t}}\sqrt{\frac{d\langle V\rangle_t}{dt}} =  \sqrt{ \frac{\left(\Sigma'_1(Y_t)\Sigma_1(.Y_t) + \Sigma'_2(Y_t)\Sigma_2(Y_t)\right)^2\sigma(Y_t)^2}{\Sigma_1(Y_t)^2 + \Sigma_2(Y_t)^2}}
\end{aligned}
\end{equation}
The proofs of Equations~\eqref{eq:VXcor} and \eqref{eq:volvol} are given in Appendix~\ref{sec:MixProofs}.

We fix a finite time horizon $T>0$. The following proposition shows that the distribution of the log price is given by a Gaussian mixture density with an infinite number of components, see \cite[Chapter~6.2]{mcneil2015quantitative}.

\begin{proposition}\label{prop:condGM}
The distribution of $X_T$ conditional on the trajectory of $W_{1t}$ on $[0,T]$ is normally distributed with mean
\begin{equation} \label{eq:MT}
M_T = X_0 + (r-\delta)T - \frac{1}{2} \int_0^T \left(\Sigma_1(Y_t)^2 + \Sigma_2(Y_t)^2\right)dt + \int_0^T \Sigma_1(Y_t)dW_{1t}
\end{equation}
and variance
\begin{equation} \label{eq:CT}
C_T = \int_0^T \Sigma_2(Y_t)^2dt.
\end{equation}
\end{proposition}

Therefore, when it is well defined, the log price density is of the form
\begin{equation} \label{eq:gGMdens}
g(x) = \E\left[ (2\pi C_T)^{-\frac{1}{2}} \exp\left(-\frac{(x - M_T)^2}{2C_T}\right) \right].
\end{equation}
Similar expressions for the log price density have been previously derived in~\cite{lipton2008stochastic} and in~\cite{glasserman2011gamma}.
In practice, we want to approximate the log price density $g(x)$ by a Gaussian mixture density $w_K(x)$ with $K$ components that will in turn be used as auxiliary density to derive option price approximations,
\begin{equation} \label{eq:GM}
w_K(x) = \sum_{k=1}^K \; c_k \,  \frac{1}{\sqrt{2\pi \sigma_{k}^2}} \exp\left(-\frac{(x - \mu_k)^2}{2\sigma_k^2}\right)
\end{equation}
for some mixture weights $c_k>0$ such that $\sum_{k=1}^K c_k=1$, and some constants $\mu_k\in\R$ and $\sigma_k>0$ for $k=1,\dots,K$.
In Section~\ref{sec:GMcons} we suggest a computationally efficient approach based on weighted simulations of the first Brownian motion $W_{1t}$.

We now give conditions such the likelihood ratio $\ell(x)$ belongs to the weighted Lebesgue space $L_{w_K}^2$.

\begin{proposition}\label{prop:GMparam}
Assume that there exist two constants $K_1,K_2>0$ such that
\begin{equation}\label{eq:MCbounded}
|M_T|<K_1 \quad \text{and} \quad C_T<K_2.
\end{equation}
If $\sigma_k > K_2/2$ for some $k$ in $1,\dots,K$ then $\ell \in L_{w_K}^2$ with $w_K$ as defined in~\eqref{eq:GM}.
Property~\eqref{eq:MCbounded} holds if $Y_t$ takes values in some compact interval $I\subset\R$ for all $t\ge 0$ and $\Sigma_1(y) =\phi(y) \sigma(y) $ for $y\in I$ for some $\phi\in C^1(I)$.
\end{proposition}

\subsection{Greeks and Parameter Sensitivity}

The option Greeks are computed by differentiating the option prices with respect to specific model parameters.
For the sensitivity analysis we fix the auxiliary density $w(x)$, hence the basis $H_n(x)$ and the coefficients $f_n$, and let only $\ell(x)$ through $g(x)$ depend on the perturbed parameters.
The $m$-th order sensitivity of $\pi_f$ with respect to the parameter $\theta$ is hence given by
\begin{equation}\label{eq:greeks}
{\partial^m_\theta \pi_f } = \sum_{n\ge 0} f_n {\partial^m_\theta \ell_n }, \quad m  \in \N,
\end{equation}
as $\partial_\theta f_n = 0$. The sensitivity of $\ell_n$ with respect to $\theta$ is given by
\begin{equation} \label{eq:delln}
{\partial_\theta \ell_n }  = \partial_\theta Q(Z_0) \,  \e^{G\, T} \, \vec{H_n} + Q(Z_0) \, {\partial_\theta \e^{G \, T}} \, \vec{H_n}.
\end{equation}
Note that the option Greeks with respect to multiple parameters can be represented as in~\eqref{eq:greeks} with mixed derivatives.

\begin{example}\label{exa:delgam}
The Delta $\Delta_f$ and the Gamma $\Gamma_f$ of an option are obtained by differentiating its price $\pi_f$ with respect to the underlying initial price $\theta=\exp(X_0)$, once and twice, respectively.
In this case, we have that $\partial_\theta \e^{G \, T}=0$ and the Greeks simplify to
\[
\Delta_f =\sum_{n\ge 0} \e^{-X_0} \frac{\partial Q(Z_0)}{\partial X_0}   \e^{G\, T} \, \vec{H_n}
\quad
\text{and}
\quad
\Gamma_f = - \e^{-X_0} \Delta_f + \sum_{n\ge 0} \e^{-2X_0} \frac{\partial^2 Q(Z_0)}{\partial X_0 \partial X_0}  \e^{G\, T} \, \vec{H_n}.
\]
Note that the derivatives of $Q(Z_0)$ with respect to $X_0$ are trivial to compute as $Q(Z_0)$ is a vector of polynomials in $X_0$ and $Y_0$.
\end{example}

To derive the option price sensitivity with respect to some model parameter, we typically need to compute the derivative of the exponential operator $\e^{G\,T}$ with respect to $y$, which is also given by
\begin{equation}\label{eq:deGT}
\frac{\partial \e^{G \, T}}{\partial \theta } = \int_0^1 \e^{x \, G \, T} \; \frac{\partial \, G \, T}{\partial y} \; \e^{(1-x)\, G \, T}\,dx
\end{equation}
as proved in \cite{wilcox1967exponential}.
However, the direct numerical integration in~\eqref{eq:deGT} may be difficult in practice.
Fortunately, the complex-step approach described in~\cite{al2010complex} turns out to be an efficient and precise way to approximate the derivative of the matrix exponential, it is shown that
\begin{equation}\label{eq:deGTapp}
\frac{\partial \e^{G \, T}}{\partial \theta } =  \Im \left(\frac{ \exp(GT + \im h E)}{h}\right) + \mathcal{O}(h^2)
\end{equation}
where $E=\partial GT / \partial \theta$, $|h| \ll 10^{-4}$ is a small constant, and where $\Im$ and $\im$ denote the imaginary part and imaginary number respectively.

The following result shows how to compute~\eqref{eq:deGTapp} without making use of complex arithmetic.
\begin{proposition} \label{prop:deGTv}
We have that
\begin{equation}
\Im \left(\frac{ \exp(GT + \im h E)}{h}\right) = \frac{1}{h} \exp\begin{pmatrix}
GT & - hE \\
hE & GT
\end{pmatrix}_{I,J}
\end{equation}
with the set of indices $I={M+2,\dots,2M+2}$ and $J={1,\dots,M+1}$.
\end{proposition}
With this result, one can also apply efficient algorithms to directly compute the action of the matrix exponential in~\eqref{eq:delln} as follows
\begin{equation}\label{eq:deAtv}
{\partial_\theta \e^{G \, T}} \, \vec{H_n} \approx \frac{1}{h} \left( \exp\begin{pmatrix}
GT & - hE \\
hE & GT
\end{pmatrix} \begin{pmatrix}
\vec{H_n} \\ 0_{M+1}
\end{pmatrix}\right)_{1, I}.
\end{equation}
In practice, we can efficiently approximate the option price sensitivities by truncating the series~\eqref{eq:greeks} at some finite order $N$.
In Section~\ref{sec:calib} we apply the above results to compute the Greeks and the gradient of the squared price error in the Heston model.

\subsection{Gaussian Mixture Specification} \label{sec:GMcons}

We present an efficient approach to specify Gaussian mixtures for the auxiliary density.
It is based on the discretization of the single source of randomness affecting $M_T$ and $C_T$ in \eqref{eq:MT} and \eqref{eq:CT}: the trajectory of $W_1$ on $[0,T]$.
Fix a time grid $0=t_0<t_1<t_2<\dots<t_n=T$ with constant step size $t_{i+1}-t_i=\Delta t$.
Then, apply a Monte-Carlo approach to obtain $K$ weighted $\R^n$-valued vectors $Z^{(k)}$ of normally distributed Brownian increments
\begin{equation}\label{eq:Zk}
Z^{(k)} \sim \left(\Delta W_{1,t_1}, \, \dots, \, \Delta W_{1,t_n} \right) \sim \Ncal\left(\mathbf{0}_n, \, \diag(\Delta t \,  \mathbf{1}_n)\right)
\end{equation}
where $\Delta W_{1,t_i}=W_{1,t_{i}} - W_{1,t_{i-1}}$.
The stochastic differential equations satisfied by $(Y_t, \, M_t, \, C_t)$ can then be numerically integrated to obtain $K$ triplets $(c_k, \,M_T^{(k)}, \, C_T^{(k)})$ where $c_k$ is the weight associated to $Z^{(k)}$.
We can then let these triplets be the parameters $(c_k, \, \mu_k, \, \sigma_k^2)$ of the Gaussian mixture in~\eqref{eq:GM}.

Standard i.i.d.\ simulations such that $c_k=1/K$ would be costly because the number of triplets $K$ required to obtain a good approximation of $g(x)$ may be large.
In addition, a large $K$ makes the computation of the coefficients $f_n$ and $\ell_n$ computationally more demanding.
Therefore, a weighted Monte-Carlo method appears to be a more sensible approach to parametrize the Gaussian mixture.
For example, one can use an optimal $K$-quantization of the multivariate Gaussian distribution.
This is a discrete probability distribution in $\R^n$ with $K$ mass points that best approximates the multivariate normal~\eqref{eq:Zk} in the $L^2$ sense, we refer to~\cite{pages2003optimal} for more details.
The $K$ pairs $(c_k,\,Z^{(k)})$ can always be precomputed and loaded on demand because they do not depend on the stochastic volatility model at hand.

The discretization scheme used to numerically integrate the stochastic differential equations may also play an important role.
We want a scheme that performs well with a large time step $\Delta t$ as we are only interested in the log price density at time $T$.
Numerical experiments on multiple models showed that good results can notably be obtained with the Interpolated-Kahl-J\"ackel (IJK) scheme introduced in~\cite{kahl2006fast}.
The IJK scheme boils down to a scheme with linear interpolation of the drift and of the diffusion, consideration for the correlated diffusive terms, and with a higher order Milstein term.
It bears little additional computational cost.

\begin{remark}[Moment matched auxiliary density]
An alternative way to construct an auxiliary density is to search for the component weights and parameters such that the first $N^*$ moments of the log price and of the auxiliary distribution match, for some fixed integer $N^*$.
This problem typically has many solutions when $K$ is large and no solution when $K$ is small.
Take the density of an arbitrary solution of this moment problem, there is a priori no guarantee that it has the same shape as the log price density.
This may lead to erroneous option price approximations as the misspecified auxiliary density also impacts the payoff coefficients.
In our experience, it is difficult to regularize or constraint the moment matching problem so as to obtain a practical auxiliary density.
\end{remark}

\subsection{Nonconvergent Option Price Approximations} \label{sec:divapprox}

We provide some explanations why, even when $\ell \notin L_{w}^2$, the price approximation $\pi_f^{(N)}$ in~\eqref{eq:piN} can be an accurate approximation of the true option price $\pi_f$.
More precisely, we estimate the option price approximation error.

Assume that instead of considering the true process $X_t$ we consider the modified process $X_t^\tau$ whose dynamics is
\[
dX^\tau_t = (r-\delta)dt - \frac{1}{2}d\langle X^\tau\rangle_t + 1_{\{\tau>t\}}\left(\Sigma_1(Y_t)\, dW_{1t} + \Sigma_2(Y_t)\, dW_{2t}\right)
\]
where the stopping time $\tau$ is defined by
\[
\tau = \inf\left\{ t\ge0 : |M_t| \geq K_1 \text{ or } C_t \geq K_2  \right\}
\]
for some positive constants $K_1$ and $K_2$.
The event $\{\tau<T\}$ is thought to be very unlikely, so that ${\Pa[X_T=X_T^\tau]} > 1-\epsilon_0$ for some small $\epsilon_0>0$.
Note that the modified discounted cum-dividend asset price process $\e^{-(r-\delta)t + X^\tau_t}$ remains a martingale.
As a consequence of this construction, Proposition~\ref{prop:GMparam} applies to the process $X_t^\tau$ which implies that the option price has a series representation for some Gaussian mixture density $w(x)$,
\begin{equation}\label{eq:pricestopped}
\pi^\tau_f=\E[f(X^\tau_{T})]=\sum_{n=0}^\infty f_n\ell_n^{\tau},
\end{equation}
where the payoff coefficients $f_n$ are defined as in~\eqref{eq:fn} and the likelihood coefficients $\ell^{\tau}_n$ are defined by
\begin{equation}\label{eq:defHermitestopped}
\ell_n^{\tau}=\E[H_n(X^\tau_{T})].
\end{equation}
However, we do not know the moments of $X_T^{\tau}$ and we will use instead the moments of $X_T$.
We hence obtain the option price approximation
\begin{equation*} \label{eq:piN2}
\pi_f^{(N)}=\sum_{n=0}^Nf_n\ell_n,
\end{equation*}
which is similar to~\eqref{eq:piN} but with the important difference that it will not converge to the true price $\pi_f$.
The pricing error with this approximation can be decomposed into three terms,
\begin{equation}\label{eqerrort}
  \left| \pi_f - \pi_f^{(N)} \right|  \leq  \left|\E[(f(X_T)-f(X_t^\tau))1_{\{\tau\le T\}}]\right|
+ \left|  \pi^\tau_f - \sum_{n=0}^N f_n \ell_n^\tau \right|
+ \left| \sum_{n=0}^N f_n (\ell_n-\ell_n^\tau) \right| =: \epsilon_1+\epsilon_2+\epsilon_3.
\end{equation}
The first error term $\epsilon_1$ in \eqref{eqerrort} is the difference between the option price with $X_T$ and with $X^\tau_{T}$.
It can easily be controlled. For example if $|f(x)|$ is bounded by $K$ on $\R$ then we have
\[\epsilon_1 \le 2 \, K \,\epsilon_0.\]
The difference between the call option prices can also be bounded by using the put-call parity given that $\E[\e^{X_T}]=\E[\e^{X^\tau_T}]$.
The second error term $\epsilon_2$ in \eqref{eqerrort} is the option price approximation error for the log price $X^\tau_T$ that converges to zero as $N\to\infty$.
The third error term $\epsilon_3$ in \eqref{eqerrort} is the difference between the option price approximations of $X_T$ and $X_T^\tau$, that we expect to be small for a low order $N$ but will typically diverge as $N\to\infty$.
This last error term can be further decomposed as
\begin{equation}\label{eq:thirderror}
\begin{aligned}
\epsilon_3 &=\left|\E\left[\sum_{n=0}^Nf_n(H_n(X_T)-H_n(X^{\tau}_{T}))1_{\{\tau\leq T\}}\right]\right|  \\
&\le \left|\E\left[\sum_{n=0}^Nf_nH_n(X_T)1_{\{\tau\leq T\}}\right]\right|+\left|\E\left[\sum_{n=0}^Nf_nH_n(X^{\tau}_T)1_{\{\tau\leq T\}}\right]\right|.
\end{aligned}
\end{equation}
The first term on the right side of~\eqref{eq:thirderror} is precisely the one that is not expected to converge as $N\to\infty$.
Applying the Cauchy-Schwarz inequality we derive the upper bound
\[
\left|\E\left[\sum_{n=0}^Nf_nH_n(X_T)1_{\{\tau\leq T\}}\right]\right|\le \sqrt{\E[p_N(X_T)] \, \epsilon_0},
\quad \text{where} \quad p_N (x)=\left(\sum_{n=0}^Nf_nH_n(x) \right)^2
\]
which can be computed explicitly.
In practice, this bound may give an indication to whether the option price approximation $\pi^{(N)}_f$ is reasonable.

\begin{remark} \label{rem:momentmatching}
A simple trick to stabilize the option price approximation when Proposition~\ref{prop:GMparam} does not apply is to match the $N^*$-th moment of the log price and the auxiliary density, $\E[X_T^{N^*}] = \int_\R x^{N^*} w(x)dx$.
By doing so, the first $N^*$ moments of $g(x)$ and $w(x)$ will be of the same magnitudes which should result in well-behaved values for the $\ell_n$ coefficients.
This can be achieved, for example, by including a component with small mixture weight and which is used to tune the $N^*$-th moment.
\end{remark}

\section{Extensions} \label{sec:ext}

In this section we discuss how multivariate expansions can be constructed by letting the auxiliary density be a tensor product of univariate densities.
We also present weighted least squares methods to approximate multivariate payoff coefficients as an alternative approach as the derivation of explicit recursive formulas in higher dimensions may be challenging.

\subsection{Multivariate expansions}

In this section we adapt our approach to tackle multivariate pricing problems by working with a tensor product of univariate auxiliary densities, in which case the multivariate basis is also a tensor product of univariate bases.

Let $\mathbf X$ be a $d$-dimensional random variable with density $g(\mathbf x)$.
The coordinates of $\mathbf X$ can be log prices of multiple assets or of the same assets at multiple dates.
We aim to compute the option price
\[
\pi_f = \int_{\R^d} f(\mathbf x) g(\mathbf x) d \mathbf x
\]
for some discounted payoff function $f(\mathbf x)$.
We transform $\mathbf X$ in order to work with
\[
\mathbf Y = A \mathbf X + b
\]
for some $(d\times d)$-dimensional non-singular matrix $A$ and some $d$-dimensional vector $b$.
The pricing problem then rewrites
\[
\pi_f = \int_{\R^d} \tilde f(\mathbf y)  \tilde g(\mathbf y) d \mathbf y
\]
where $\tilde g (\mathbf y)= g(A^{-1}(\mathbf y - b))$ is the density of $\mathbf Y$, and $\tilde f(\mathbf y)= f(A^{-1}(\mathbf y - b))$ is the new payoff function.
The coordinates of $\mathbf Y$ can be made (conditionally) uncorrelated in which case it is sensible to let the auxiliary density be a tensor product of univariate densities, that is
\begin{equation}\label{eq:tensorw}
w(\mathbf y) = \prod_{i=1}^d w_i(y_i)
\end{equation}
for some univariate densities $w_i(y_i)$ with $i=1,\dots,d$.
The multivariate orthogonal basis is then also a tensor product of the marginal orthogonal basis and the marginal univariate densities $w_i(y_i)$ can be  specified arbitrarily.
For example~\cite{ackerer2018jacobi} approximated a forward-start option price by working with the log returns in place of the log prices and using a tensor product of Gaussian densities as auxiliary density.

\begin{remark}
Multivariate orthogonal polynomial expansions with respect to an auxiliary density which cannot be written as in~\eqref{eq:tensorw} is fundamentally more difficult for several reasons.
First, the three-term relation~\eqref{eq:ONBkrec} takes a vector-matrix form whose coefficients are in general not known explicitly, see~\cite[Chapter 3.3]{dunkl2014orthogonal}.
Second, Proposition~\ref{prop:ONBmix} does not hold in general for multivariate orthogonal basis.
Indeed, its proof builds on the Gauss quadrature rule associated to each univariate density whose multivariate equivalent does not necessarily exists, see~\cite[Chapter 3.8]{dunkl2014orthogonal}.
Third, the basis dimensionality may create computational challenges in practice.
\end{remark}

\subsection{Payoff coefficients approximation} \label{sec:wls}

The payoff coefficients may not be tractable for some payoff functions or for some auxiliary densities.
In such situation, one could still integrate the payoff function with respect to the approximated density as presented in~\cite{filipovic2013density}. For example~\cite{ackerer2018jacobi} approximated a discretely monitored Asian option price using multivariate Gaussian quadrature.
Alternatively, we discuss hereinafter a generic way to approximate the payoff coefficients which may lead to further reduction of computing time.
More precisely, this approach belongs to the class of discrete weighted least squares methods studied in Computational Mathematics, see~\cite{cohen2013stability,hampton2015compressive,narayan2017christoffel}.

Let $H_n(\mathbf x)$ be an orthogonal polynomial basis with respect to the possibly multivariate auxiliary density $w(\mathbf x)$.
Fix an integer $N$ and let $M$ be the number of basis elements with total order less or equal than $N$.
We approximate the first $M$ payoff coefficients $f_n$ of the payoff function $f(\mathbf x)$ by solving the following weighted least-square problem
\begin{equation}\label{eq:wlspb}
\argmin_{\mathbf c \in \R^M} \;  \frac{1}{N_{\rm sim}} \, \sum_{i=1}^{N_{\rm sim}} \, q_i \, \left(f(\mathbf{x}_i) - \sum_{|n|\le N} c_n\,H_n(\mathbf{x}_i) \right)^2
\end{equation}
for some $N_{\rm sim}$ astutely simulated points $\mathbf{x}_i$ with corresponding weight $q_i$.
By expanding this expression we have that a solution $\mathbf c$ of~\eqref{eq:wlspb} should satisfy
\[
\mathbf G \mathbf c = \mathbf d
\]
where the $M\times M$ Gramian matrix $\mathbf G$ and the $M$ vector $\mathbf d$ are defined by
\[
\mathbf G_{j,k} = \frac{1}{N_{\rm sim}} \, \sum_{i=1}^{N_{\rm sim}} \, p_i \,  H_j(\mathbf{x}_i)  \, H_k(\mathbf{x}_i)
\quad
\text{and}
\quad
\mathbf d_j = \frac{1}{N_{\rm sim}} \, \sum_{i=1}^{N_{\rm sim}} \, p_i \, H_j(\mathbf{x}_i) \, f(\mathbf{x}_i).
\]
If $\mathbf G$ is non-singular then~\eqref{eq:wlspb} has a unique solution.
The exact payoff coefficients $f_n$ are solution of the following problem
\[
\argmin_{\mathbf c \in \R^M} \;  \Big\Vert f(\mathbf{x}) - \sum_{|n|\le N} c_n\,H_n(\mathbf{x}) \Big\Vert^2_w
\]
which makes it tempting to simulate the points $\mathbf{x}_i$ from the density $w(\mathbf{x})$ which implies that $q_i=1$ and that~\eqref{eq:wlspb} simplifies to a least-square problem.
However, this does not guarantee that the matrix $\mathbf G$ is invertible and would typically require a large number of simulations $N_{\rm sim}$.
Efficient weighted simulation methods are therefore being developed to tackle these issues, see for example~\cite{cohen2017optimal}.
One limitation of this approach is that all the coefficients with total order less or equal than $N$ are computed at once rather than incrementally.
The simulated random variables $\mathbf{x}_i$ do not depend on the particular payoff and can therefore be used to approximate the payoff of many contracts with the same maturity.
Note that these random variables are significantly faster to simulate than the log price paths.

\section{Numerical Applications} \label{sec:MIXapp}

In this section we approximate the price of call options in the Jacobi, Stein--Stein, and Hull-White models, and of an option on the realized variance in the GARCH model.
We then approximate call option Greeks for the Heston model, and benchmark the computational efficiency of our approach against the Fourier transform formulas.
We also approximate the payoff coefficients using a particular weighted least squares algorithm.

\subsection{Option Pricing} \label{sec:num_pricing}

We show that a simple Gaussian mixture with two components can be used to improve the convergence rate of the option price approximations in the Jacobi model.
We then derive accurate option price approximations in the Stein--Stein and Hull--White models using Gaussian mixtures with many components for the auxiliary density.
We also approximate the price of an option on the realized variance in the GARCH diffusion model.

\paragraph{The Jacobi model.}
The dynamics of $(X_t,Y_t)$ in the Jacobi model has the form
\begin{align*}
dY_t &= \kappa(\theta-Y_t)\,dt + \sigma\sqrt{Q(Y_t)}\,dW_{1t}\\
dX_t &= \left(r-\delta -  Y_t/2\right)dt +  \rho\sqrt{Q(Y_t)} \, dW_{1t}+\sqrt{Y_t-\rho^2 Q(Y_t)} \,dW_{2t}
\end{align*}
where $Q(y)={(y-y_{min})(y_{max}-y)}/{(\sqrt{y_{max}}-\sqrt{y_{min}})^2}$, $\kappa>0$, and $\theta\in(v_{min},v_{max}]$.
The log return volatility $\sqrt{Y_t}$ takes values in $[\sqrt{v_{min}},\,\sqrt{v_{max}}]$, the leverage is ${\rm lev}(X_t)=\rho \sqrt{{Q(Y_t)}/{Y_t}}$, and the volatility of volatility is ${\rm volvol}(X_t)=(\sigma/2)\sqrt{{Q(Y_t)}/{Y_t}}$.
The Jacobi model was introduced in~\cite{ackerer2018jacobi}.

We illustrate the advantages of using a Gaussian mixture with the Jacobi model when ${\mathbb{V}{\rm ar}}[X_T]<y_{max}T/2$, so that the first two moments cannot be matched with a single auxiliary Gaussian density $w_1(x)=v_1(x)$.
We consider here a Gaussian mixture with two components $w_2(x)$ as auxiliary density given by
\[
w_2(x) = c_1 \, \frac{1}{\sqrt{2\pi\sigma_{1}^2}}\exp\left(-\frac{(x - \mu_{1})^2}{2\sigma_{1}^2}\right) + (1-c_1)  \, \frac{1}{\sqrt{2\pi\sigma_{2}^2}}\exp\left(-\frac{(x - \mu_{2})^2}{2\sigma_{2}^2}\right)
\]
for some weight $0<c_1<1$, some mean parameters $\mu_{1},\mu_{2}\in\R$ and volatility parameters $\sigma_{1},\sigma_{2}>0$.
We match the the first two moments of $X_T$ and $w_2(x)$, which gives the following underdetermined system of equations
\begin{align}
\E[X_T] &= c_1\,\mu_{1} + (1-c_1)\,\mu_{2} \label{eq:mts1}\\
\E[X_T^2] &= c_1\,(\sigma_{1}^2 + \mu_{1}^2) + (1-c_1)\,(\sigma_{2}^2 + \mu_{2}^2)\label{eq:mts2}.
\end{align}
We set $\E[X_T] = \mu_{1} = \mu_{2}$, so that \eqref{eq:mts1} is automatically satisfied and \eqref{eq:mts2} rewrites
\begin{equation}\label{eq:mts2bis}
c_1 = \frac{\sigma_{2}^2 - {\mathbb{V}{\rm ar}}[X_T]}{\sigma_{2}^2 - \sigma_{1}^2}.
\end{equation}
Hence it must be that
$\left\vert \sigma_{2}^2 - {\mathbb{V}{\rm ar}}[X_T] \right\vert < \left\vert \sigma_{2}^2 - \sigma_{1}^2 \right\vert$
and that
$\left( \sigma_{2}^2 - {\mathbb{V}{\rm ar}}[X_T] \right)\left( \sigma_{2}^2 - \sigma_{1}^2 \right)\ge0$.
We set $\sigma_2=\sqrt{y_{max}T/2}+10^{-4}$, so that Proposition~\ref{prop:GMparam} applies which ensures that the option price expansion will converge to the true price.
Then we arbitrarily fix $c_1=95\%$ and solve~\eqref{eq:mts2bis} to get $\sigma_{1}$.
The reason behind these choices is that, by doing so, the mixture component with large weight $c_1$ is almost a Gaussian approximation of the log price density as $\E[X_T]=\mu_{1}$ and ${\mathbb{V}{\rm ar}}[X_T]\approx\sigma_{1}^2$.
In the following numerical example we use the parameters: $r=\delta=x_0=0$, $\kappa=0.5$, $\theta=Y_0=0.04$, $\sigma=1$, $y_{min}=10^{-4}$, $y_{max}=0.36$, and $T=1/12$.
The upper bound on the volatility support is therefore $60\%$.

The left panel of Figure~\ref{fig:IV-GM} displays the Gaussian mixture with one and two components used as auxiliary density along with the log price density approximation~\eqref{eq:densapp} at the truncation order $N=100$.
We expect that $g(x) \approx \ell^{(100)}(x)w_2(x)$.
It is clear from the figure that $w_1(x)$ is a poor approximation of the density $g(x)$ whereas $w_2(x)$ appears more sensible.
As a consequence, the implied volatility of a call option with log strike ${\bar k}=0$ converges significantly faster using $w_2(x)$ as auxiliary density as can be seen on the right panel of Figure~\ref{fig:IV-GM}.
The call (put) implied volatility is initially overestimated with $w_1(x)$ because significantly more weight is put in the right (left) tail than with $g(x)$.
This behavior is confirmed in Table~\ref{tab:IV-GM} which reports the implied volatility error with respect to the approximation at the order $N=100$ for call options with different moneyness and for different truncation order.
More problematic, the option price approximation of the far OTM option is negative between $N=2$ and $N=18$ with $w_1(x)$.

\paragraph{The Stein--Stein model.}
The dynamics of $(X_t,Y_t)$ in the Stein--Stein model has the form
\begin{equation}\label{eq:linearSV}
\begin{split}
dY_t &= \kappa(\theta - Y_t) dt + \sigma dW_{1t} \\
dX_t &= \left(r-\delta -  Y_t^2/2\right)dt + \rho Y_t dW_{1t} + \sqrt{1-\rho^2} Y_t dW_{2t}
\end{split}
\end{equation}
for some positive parameters $\kappa$, $\theta$, and $\sigma$, and with $\rho\in(-1,1)$.
The process $Y_t$ follows an Ornstein-Uhlenbeck process.
The asset return volatility is $|Y_t|$, the leverage is ${\rm lev}(X_t)=\rho \sign(Y_t)$, and the volatility of volatility is ${\rm volvol}(X_t)=\sigma$.
This model has been introduced by~\cite{stein1991stock} and generalized by~\cite{schobel1999stochastic} to allow for non-zero leverage.

The Stein--Stein model has the particularity that the diffusion $(X_t,Y_t,Y_t^2)$ is not only polynomial but also affine.
This enables the use of standard Fourier transform techniques to compute option prices that we use as reference option prices, see for details~\cite{carr1999option}, \cite{duffie2003affine}, and \cite{fang2009novel}.
We approximate option prices using a Gaussian mixture for the auxiliary density as described in Section~\ref{sec:GMcons}.

Applying the IJK scheme we obtain the following discretization for the dynamics of the process $Y_t$
\[
Y_{t_{i+1}} = Y_{t_{i}} + \kappa(\theta - Y_{t_{i}})\Delta t + \sigma \Delta W_{1,t_{i+1}}  \quad i=0,\dots,n-1
\]
for the conditional mean
\[
M_T = X_0 -\frac{1}{2} \sum_{i=1}^n \frac{Y_{t_{i}}^2+Y_{t_{i-1}}^2}{2}\Delta t + \rho \sum_{i=1}^n Y_{t_{i}} \Delta W_{1,t_i} + \frac{1}{2}\rho\sigma \sum_{i=1}^n (\Delta W_{1,t_i}^2 - \Delta t)
\]
and for the conditional variance
\[
C_T = (1-\rho^2) \sum_{i=1}^n \frac{Y_{t_{i}}^2+Y_{t_{i-1}}^2}{2}\Delta t.
\]

The time to maturity is $T=1/12$ and we use a single step, $\Delta t =T$. We therefore use the optimal quantizers of the univariate normal distribution to approximate the Brownian increment $\Delta W_{1\Delta t}$.
We also approximate option prices using an extended mixture obtained by including one additional component whose purpose is to adjust the moments of the auxiliary density as suggested in Remark~\ref{rem:momentmatching}.
This additional component has a fixed weight equal to $5\%$ while the other component weights are scaled down by $95\%$, its mean parameter is set equal to zero, and its variance parameter is computed such that the $N^*$-th moment of the auxiliary density and of the log price density are equal.
A different number of components $K$ for each auxiliary density and the first moment of the log price is always matched such that $\ell_1=0$.
The parameters are $N^*=20$, $r=\delta=X_0={\bar k}=0$, $\kappa=0.5$, $\theta=Y_0=0.2$, $\sigma=0.5$, and $\rho=-0.5$.

Table~\ref{tab:SteinIV} reports the implied volatility errors for a call option with log strike $k=0$ for different $K$ and using the Gaussian mixture and the extended Gaussian mixture as auxiliary density, respectively denoted GM and GM$^+$.
We can see that the implied volatility errors rapidly become small as the truncation order $N$ increases for all $K$ with the extended mixture and continue to converge toward $\pi_f$ well after the matched moment $N^*$.
On the other hand, the implied volatility errors do not appear to converge when $K$ is small with the standard Gaussian mixture.
Note that good option price approximations can be achieved by choosing a large number of components $K$ even without moment matching.

\paragraph{The Hull--White model.}
The dynamics of $(X_t, Y_t)$ in the Hull--White model has the form
\begin{equation}\label{eq:HW}
\begin{split}
dY_t &= \kappa(\theta - Y_t) dt + (\nu + \gamma Y_t)dW_{1t} \\
dX_t &= (r-\delta -Y_t^2/2)dt + \rho Y_t dW_{1t} + \sqrt{1-\rho^2} Y_t dW_{2t}
\end{split}
\end{equation}
for some positive parameters $\kappa, \theta, \gamma>0$, a real-valued parameter $\nu$, and with $\rho\in(-1,1)$.
The process $Y_t$ takes values in $({-\nu}/{\gamma},\infty)$ if $\theta\ge{-\nu}/{\gamma}\ge 0$, so that $\nu + \gamma Y_t>0$.
The log return volatility is $|Y_t|$, the leverage is ${\rm lev}(X_t)=\rho $, the volatility of volatility is ${\rm volvol}(X_t)=\nu + \gamma Y_t$.

The Hull--White model has been introduced by~\cite{hull1987pricing} with $\nu=0$ and $\rho=0$ and has been generalized by \cite{sepp2016log} for non-zero leverage and jumps but still with $\nu=0$.
The Stein--Stein model is obtained as a limit case for $\gamma\rightarrow 0$ and $\nu\neq0$.
The dynamics of the Hull--White model is extremely volatile as, for example, the stock price does not even admit a second moment for any finite horizon when the correlation parameter $\rho$ is not sufficiently negative, see \cite{lions2007correlations}.

The diffusion $(Y_t, Y_t^2, X_t)$ has the polynomial property but is not affine.
Indeed, applying the infinitesimal generator of \eqref{eq:HW} to a monomial $Y_t^mX_t^n$ gives
\begin{equation}\label{eq:linSVgen}
\begin{split}
\Gcal\, y^m x^n &= m (\kappa \theta + (m-1)\nu\gamma) y^{m-1}x^{n} + m(0.5(m-1)\gamma^2 - \kappa)y^{m}x^{n} \\
&\quad + 0.5 m(m-1) \nu^2 y^{m-2} x^{n} + m n \rho \gamma y^{m+1}x^{n-1} +n(m\rho\nu + r-\delta)y^{m}x^{n-1} \\
&\quad + 0.5 n(n-1) y^{m+2}x^{n-2} - 0.5 n y^{m+2}x^{n-1}
\end{split}
\end{equation}
for any integers $m$ and $n$. Note that the last term in~\eqref{eq:linSVgen} is a monomial of order $m+n+1$, so that the diffusion $(Y_t,X_t)$ does not have the polynomial property on its own.

Applying the IJK scheme we obtain the following discretization for the dynamics of the process $Y_t$
\[
Y_{t_{i+1}} = Y_{t_{i}} + \kappa(\theta - Y_{t_{i}})\Delta t + (\nu+\gamma Y_{t_{i}}) \Delta W_{1,t_{i+1}} + \frac{1}{2} \gamma (\nu+\gamma Y_{t_{i}})(\Delta W_{1,t_{i+1}}^2 - \Delta t), \quad i=0,\dots,n-1
\]
for the conditional mean
\[
M_T = X_0 -\frac{1}{2} \sum_{i=1}^n \frac{Y_{t_{i}}^2+Y_{t_{i-1}}^2}{2}\Delta t + \rho \sum_{i=1}^n Y_{t_{i}} \Delta W_{1,t_i} + \frac{1}{2}\rho \sum_{i=1}^n (\nu+\gamma Y_{t_{i}})(\Delta W_{1,t_i}^2 - \Delta t)
\]
and for the conditional variance
\[
C_T = (1-\rho^2) \sum_{i=1}^n \frac{Y_{t_{i}}^2+Y_{t_{i-1}}^2}{2}\Delta t.
\]

We follow the same procedure as for the Stein--Stein model and set a Gaussian mixture as auxiliary density, with and without $N^*$ moment matching.
The parameters are $r=\delta=X_0=0$, $\kappa=0.5=-\rho$, $\theta=Y_0=0.2$, $\nu=0.25$, $\gamma=0.5$, and $T=1/12$.
Table~\ref{tab:HullWhiteIV} reports the implied volatility values for a call option with log strike ${\bar k}=0$ for different $K$ and using the Gaussian mixture and the extended Gaussian mixture as auxiliary density, respectively denoted GM and GM$^+$.
As for the Stein--Stein model, the price series seems to converge faster when the number of Gaussian component $K$ is large, and to stabilize for finite $N\le N^*$ when the higher moment $N^*$ is matched.
The price series behavior beyond the matched moment $N^*$ is however not as well behaved as for the Stein--Stein model.
This seems to be a consequence of the exponentially fast growing higher moments in the Hull--White model which causes the third error term $\epsilon_3$ in \eqref{eqerrort} to rapidly grow with~$N$.

Figure~\ref{fig:HWkurtosis} displays the kurtosis of the log price, $\E[(X_T-\E[X_T])^4]/\mathbb{V}{\rm ar}[X_T]^2$, for the Hull--White model for $\nu\in\{-0.05,0,0.25,0.5\}$ and $\gamma\in\{ 0.5,1\}$, and for maturities $T$ ranging from one week and one year.
The kurtosis in the limiting Stein--Stein model is obtained for $\gamma= 0$.
We observe that the kurtosis grows with time $T$ and in particular with $\gamma$ and, for larger $\gamma$ values, the kurtosis appears to explode at some point in time.
This is in sharp contrast to the Black-Scholes model for which the kurtosis is constant, and to the Stein--Stein model for which the kurtosis range appears to take reasonable values.
The erratic moments behavior renders polynomial option pricing techniques more delicate to apply for the Hull--White model than for the Jacobi or Stein--Stein models, partly because of numerical precision difficulties.

\begin{remark}
This higher moments explosion in the Hull-White model can be related to the non-stationarity of the log price instantaneous variance process $Y_t^2$ whose dynamics is given by
\[
dY^2_t = \left((\gamma^2 - 2\kappa)Y_t^2 + 2\kappa\theta Y_t + \nu^2 \right)dt + 2\left(\nu Y_t  + \gamma Y^2_t\right) dW_{1t}
\]
where the drift is positive for all $Y_t>-\nu/\gamma$ when
\begin {enumerate*} [label=\itshape\alph*\upshape)]
\item $\gamma^2>2\kappa$ and $\Delta=(2\kappa\theta)^2 - 4\nu^2(\gamma^2 - 2\kappa)<0$,
or when \item $\gamma^2=2\kappa$, $\nu \le 0$, and $-\nu\theta\sqrt{2\kappa}>\nu^2$,
or when \item $\gamma^2>2\kappa$, $\Delta\ge0$ and $(-2\kappa\theta + \Delta)/(2\gamma^2-4\kappa)) \ge -\nu/\gamma$.
\end {enumerate*}
\end{remark}

\paragraph{The GARCH model.}
The dynamics of $(X_t, Y_t)$ in the GARCH diffusion model has the form
\begin{equation}\label{eq:garch}
\begin{split}
dY_t &= \kappa(\theta - Y_t) dt + (\nu + \gamma Y_t)dW_{1t} \\
dX_t &= (r-\delta -Y_t/2)dt + \rho \sqrt{Y_t} dW_{1t} + \sqrt{1-\rho^2} \sqrt{Y_t} dW_{2t}
\end{split}
\end{equation}
for some positive parameters $\kappa,\theta,\gamma>0$, a non-positive parameter $\nu\le0$, and with $\rho\in(-1,1)$.
The log return volatility is $\sqrt{Y_t}$, the leverage is ${\rm lev}(X_t)=\rho$, the volatility of volatility is ${\rm volvol}(X_t)=\sigma\sqrt{Y_t}/2$.

When $\nu=0$ the process $Y_t$ follows the GARCH diffusion process introduced in \cite{nelson1990arch}.
The link between discrete time and continuous GARCH processes has been studied in \cite{drost1996closing}.
The diffusion $(X_t, Y_t)$ is not polynomial (and hence not affine), unless $\rho=0$, and there is no analytic expression to price options in general. A formula to approximate option prices in the GARCH diffusion model was derived in \cite{barone2005option} for the polynomial case $\rho=0$.
In this case, our method, as applied on the Stein--Stein and Hull-White models, can also be used to accurately approximate stock option prices.
When $\rho\neq 0$ the GARCH diffusion can be approximated by a polynomial model in order to approximate option prices, see~\cite{ackerer2018poly}.

We now show how to approximate the price of an option on the realized variance with discounted payoff given by
\[ \e^{-rT}(I_T - {\bar k})^+, \quad \text{where $I_t = \frac{1}{T} \int_0^t Y_s ds$.}\]

The diffusion $(Y_t, I_t)$ has the polynomial property but is not affine. Indeed, applying its infinitesimal generator to the monomial $Y_t^mX_t^n$ gives
\begin{align*}
\Gcal\, y^m i^n &= m (\kappa \theta + (m-1)\nu\gamma) y^{m-1}i^{n} + m(0.5(m-1)\gamma^2 - \kappa)y^{m}i^{n} \\
&\quad + 0.5 m(m-1) \nu^2 y^{m-2} i^{n} + (1/T) n y^{m+1} i^{n-1}
\end{align*}
for any integers $m$ and $n$.

The process $I_t$ being lower bounded by $-\nu/\gamma$, we let the auxiliary density $w(x)$ be a Gamma density with parameter $\xi=-\nu/\gamma$.
We fix the parameters $\alpha$ and $\beta$ so that the first moment is matching and that the variance of the auxiliary density is twice large than the one from $I_t$.
This gives,
\[
\E[I_T]=\xi + \frac{\alpha}{\beta} \quad \text{and} \quad  \mathbb{V}{\rm ar}[I_T]=\frac{\alpha}{2\beta^2}
\quad \Leftrightarrow \quad
\beta = \frac{\E[I_T]-\xi}{ 2\mathbb{V}{\rm ar}[I_T]} \quad \text{and} \quad \alpha = \frac{(\E[I_T]-\xi)^2}{ 2\mathbb{V}{\rm ar}[I_T]}.
\]
We do not know whether $g/w$ belongs to $L_w^2$. This motivates the arbitrary choice of setting the second moment of $w$ be strictly larger than $\E[I_t^2]$, so that the right tail of the auxiliary density is likely to be sufficiently thick.
The parameters are $r=\delta=X_0=0$, $\kappa=0.5=-\rho$, $\theta=Y_0=0.2$, $\nu=-0.025$, $\gamma=0.5$, and $T=1/2$.

The left panel of Figure~\ref{fig:RVoption} displays the auxiliary density $w$, the density approximation $\ell^{(20)}w$ at the order $20$, and a simulated density $g_{\rm mc}$ with 100'000 paths.
The simulations were performed with the time step 1/252 and the Euler discretization scheme.
As expected, because of the second moment specification, we observe that $w$ is more dispersed than $g_{\rm mc}\approx g$.
Interestingly, the density approximation at the order 20 is visually almost perfectly aligned with the simulated density.
The right panel of Figure~\ref{fig:RVoption} displays the price series approximation of an option on the realized variance with ${\bar k}=0.2$.
We notice that the price converges rapidly and stabilizes to the same value as the Monte-Carlo estimate.

\subsection{Greeks and Calibration} \label{sec:calib}

We show that Greeks can also be accurately approximated and that parameter sensitivity analysis as well as gradient based optimization are enabled using polynomial orthogonal expansions.
Furthermore, we show that the CPU time required to approximate option prices and gradients can be even faster than using analytic expressions when sufficiently many contracts are considered.
We use the Heston model as analytic formulas are readily available for the option prices and the Greeks, see Appendix~\ref{sec:Heston}.

The dynamics of $(X_t,Y_t)$ in the \cite{heston1993closed} model has the form
\begin{equation}\label{eq:heston}
\begin{split}
dY_t &= \kappa(\theta - Y_t) dt + \sigma \sqrt{Y_t} dW_{1t} \\
dX_t &= \left(r-\delta -  Y_t/2\right)dt + \rho \sqrt{Y_t} dW_{1t} + \sqrt{1-\rho^2} \sqrt{Y_t} dW_{2t}
\end{split}
\end{equation}
for some positive parameters $\kappa$, $\theta$, and $\sigma$, and with $\rho\in(-1,1)$.
The process $Y_t$ follows a square-root diffusion.
Applying the IJK scheme we obtain the following discretization for the dynamics of the process $Y_t$
\[
Y_{t_{i+1}} = Y_{t_{i}} + \kappa(\theta - Y_{t_{i}})\Delta t + \sigma \sqrt{Y_{t_i}} \Delta W_{1,t_{i+1}} + \frac{1}{4} \rho^2 (\Delta W_{1,t_{i+1}}^2 - \Delta t)\quad i=0,\dots,n-1
\]
for the conditional mean
\[
M_T = X_0 -\frac{1}{2} \sum_{i=1}^n \frac{Y_{t_{i}}+Y_{t_{i-1}}}{2}\Delta t + \rho \sum_{i=1}^n \sqrt{Y_{t_{i}}} \Delta W_{1,t_i} + \frac{1}{4}\rho\sigma \sum_{i=1}^n (\Delta W_{1,t_i}^2 - \Delta t)
\]
and for the conditional variance
\[
C_T = (1-\rho^2) \sum_{i=1}^n \frac{Y_{t_{i}}+Y_{t_{i-1}}}{2}\Delta t.
\]

We approximate option prices using the same method as in Section~\ref{sec:num_pricing} for the Stein--Stein and Hull-White models.
We use a Gaussian mixture with 21 components and such that the 20-th moment of the auxiliary density and the log price are matching.
The reference parameters are $r=\delta=X_0=0$, $\kappa=0.5=-\rho$, $\theta=Y_0=0.04$, $\sigma=0.5$, and $T=1/12$.
We use a step size of $h=10^{-6}$ in~\eqref{eq:deAtv} to compute the parameter sensitivities.

\paragraph{Delta and Gamma.}

The upper panels in Figure~\ref{fig:greeks} display the option price, the Delta, and the Gamma for various strikes computed using the analytic formulas and the orthogonal polynomial expansions as in Example~\ref{exa:delgam}.
We observe that the exact and approximated values are almost perfectly equal, with minor errors for the Gamma at the money.
This is supported by the lower panels in Figure~\ref{fig:greeks} which show that the differences between the two approaches are in general negligible, and differ not more than few percents around the money.

\paragraph{Gradient computation.}
We fix the log strike ${\bar k}=0$.
We illustrate the gradient's behavior of the squared price error for different values of $\theta$ and $\rho$.
The squared price error is defined by
\[
\text{Squared Error} = \left((\pi(\Theta) - \pi(\tilde \Theta)^{(N)}\right)^2
\]
where $\pi(\Theta)$ denotes the option price for the reference parameters $\Theta$, and $\pi(\tilde \Theta)^{(N)}$ denotes the approximated option price for the parameters $\tilde \Theta$.
Panels on the first row of Figure~\ref{fig:parsens} display the option price for different values of $\theta$ and $\rho$.
Panels on the second row show the derivative of the squared price error with respect to $\theta$ and to $\rho$.
We observe that the parameter sensitivities are smooth, and provide the correct direction toward which the parameter should change in order to reduce the square price error.
Although in principle possible, differentiating the option price in the Heston model with respect to a given parameters appears to be difficult and not practical, see~Appendix~\ref{sec:Heston}.

\paragraph{Calibration CPU time.}

The left panel of Figure~\ref{fig:cpu} displays the total CPU time required to compute or approximate option prices for different number of call options having the same maturity.
We observe that there is an upfront CPU time cost when using orthogonal polynomial expansions which is corresponds to constructing the polynomial basis and to computing the moment coefficients.
However, once these elements are available, approximating an additional option price boils down computing an additional series of payoff coefficients which is done swiftly.
Therefore, if one needs to compute more than 100 option prices of a given maturity, then the orthogonal polynomial expansion approach is at least as efficient as using the Fourier transform formula for the Heston model.

The right panel of Figure~\ref{fig:cpu} displays the total CPU time required to approximate the option price sensitivity with respect to $\kappa$, $\theta$, $\sigma$, and $\rho$ for different number of call options having the same maturity.
More precisely, we compare the approximation formula~\eqref{eq:greeks} with the finite difference method for which prices are computed using the Fourier transform formula in Appendix~\ref{sec:Heston}.
The orthogonal polynomial expansion approach performs better than the finite difference method for the Heston model when working with more than 30 options and, in addition, provide by design an accurate evaluation of the gradient.

\subsection{Payoff Coefficients Approximation}

We show that the weighted least squares approach, as discussed in Section~\ref{sec:wls}, can accurately approximate the payoff coefficients and can thus be a workable solution for option pricing applications.
We used the same configuration as in~\citep{ackerer2018jacobi} for easier comparison, that is, a Jacobi model with the parameters: $r=\delta=x_0=0$, $\kappa=0.5$, $\theta=Y_0=0.04$, $\sigma=1$, $y_{min}=10^{-4}$, $y_{max}=0.08$.
The call option has maturity $T=1/12$ and log strike ${\bar k}=0$.
The forward start call option starts at $t_1=1/12$, matures at $t_2=5/12$, and its log strike is ${\bar k}=0$.

Panels on the first row of Figure~\ref{fig:payoffapp} show the exact and approximated payoff coefficients for a call option and a forward start call option.
The auxiliary densities are respectively a Gaussian density with parameters $\mu_w=\E[X_T]$ and $\sigma^2_w={\rm Var}[X_T]$, and a tensor product of two Gaussian densities with parameters $\mu_{w_i}=\E[X_{t_i}-X_{t_{i-1}}]$ and $\sigma_{w_i}={\rm Var}[X_{t_i}-X_{t_{i-1}}]$ with $t_0=0$.
To approximate the payoff coefficients we used the optimal weighted least square method described in~\cite{cohen2017optimal} with $10^5$ simulated random variables.
We observe that all the payoff coefficients with an absolute value larger than about $10^{-6}$ are well approximated for the call option and for the forward start call option.
There are typically many payoff coefficients with negligible weights in multivariate applications, meaning that their values are so small that they do not affect the option price approximation to a significant extend.
Indeed, on the second row of Figure~\ref{fig:payoffapp} we observe that the approximated option prices computed with the exact and the approximated payoff coefficients are almost indistinguishable.

\section{Conclusion}\label{sec:ccl}

We derived tractable option price series representations for polynomial stochastic volatility models when the auxiliary density is a mixture density. In particular, we presented efficient methods to specify an auxiliary Gaussian mixture density and indicated that accurate option price and Greek approximations are possible even if the likelihood ratio function of the auxiliary density does not belong to the corresponding weighted Lebesgue space. We provided several numerical examples that illustrate the good performance of our approach. We also discussed the extension of our approach to pricing problems in higher dimensions.

\newpage
\begin{appendix}

\section{Proofs} \label{sec:MixProofs}

This Appendix contains the proofs of all theorems and propositions in the main text.

\subsection*{Proof of Proposition~\ref{prop:ONBmix}}
This proof is based on the results in~\cite{fischer1992generate}.
We first aim to derive a series of orthogonal monic polynomial basis $\{h_n\}_{n\ge 0}$, that is whose leading order coefficient is equal to one, and then normalize it to obtain the desired basis $\{H_n\}_{n\ge 0}$.
The recurrence relation for the orthogonal monic basis is given by
\begin{equation} \label{eq:ONBmonic}
x h_{n}(x) = h_{n+1}(x) + \alpha_n h_{n}(x) + \gamma_n h_{n-1}(x)
\end{equation}
for all $n\ge 0$ with $h_{-1}=0$ and $h_{0}=1$, and where the coefficients $\alpha_n,\gamma_n$ are given by
\begin{equation} \label{eq:ONBmonic:coefs}
\alpha_n = \frac{ \langle h^*_n, h_n \rangle_w }{\langle h_n, h_n \rangle_w} \quad \text{and} \quad
\gamma_n = \frac{ \langle h_n, h_n \rangle_w }{\langle h_{n-1}, h_{n-1} \rangle_w}
\end{equation}
for $n\ge 0$, and where $h^*_n(x)=x h_n(x)$.
The orthonormal polynomial basis is then obtained by normalizing the orthogonal monic basis, that is
\[
H_n(x) = \frac{h_n(x)}{\sqrt{\langle h_n, h_n \rangle_w }},
\]
which in view of Equation~\eqref{eq:ONBmonic} is equivalent to define the recurrence coefficients as follows
\[
a_{n} = \alpha_n \quad \text{and} \quad b_n = \sqrt{\gamma_n}.
\]

The inner products in Equation~\eqref{eq:ONBmonic:coefs} are left to be computed.
We show that one can actually compute effectively and accurately the integral $\langle p, 1 \rangle_w$ for any polynomial of order less than $2N$.
First, we recall the Gauss quadrature rule associated with the density $v_k$
\[
\langle p, 1 \rangle_{v_k} = \int_{\R} p(x)v_k(x)dx = \sum_{i=0}^N (\nu_{i1}^k)^2 \, p(\lambda_i^k)
\]
where $\nu_{i}^k$ is the eigenvector corresponding to the eigenvalue $\lambda_i^k$ of the Jacobi matrix $J_N^k$.
Those values however do not need to be computed explicitly.
Observe that, the matrix $J_N^k$ being Hermitian, there exists a unitary matrix $U_N^k$ whose columns are the normalized eigenvectors of $J_N^k$ and such that
\[
\Sigma_N^k := \diag(\lambda_0^k, \dots, \lambda_N^k) = (U_N^k)^\top J_N^k \, U_N^k.
\]
Combining the above results we obtain
\begin{align*}
\langle p, 1 \rangle_{w} & = \sum_{k=1}^K c_k \, \sum_{i=0}^N (\nu_{i1}^k)^2 \, p(\lambda_i^k) = \sum_{k=1}^K c_k \, \e_1^\top \, U_N^k \, p(\Sigma_N^k) (U_N^k)^\top \,  \e_1 \\
&= \sum_{k=1}^K c_k \, \e_1^\top \, p(J_N^k) \, \e_1.
\end{align*}
Define the vector $z_n^k$ as follows
\begin{align*}
z_{n+1}^k &= h_{n+1}(J_N^k)\e_1 = (J_N^k - \alpha_n)h_n(J_N^k) \e_1 - \gamma_n h_{n-1}(J_N^k)\e_1 \\
&=(J_N^k - \alpha_n)z_n^k \e_1 - \gamma_n z_{n-1}^k \e_1
\end{align*}
where the second equality follows from Equation~\eqref{eq:ONBmonic}.
The inner products then rewrites
\begin{align*}
\langle h_{n}, h_{n} \rangle_{w} &= \sum_{k=1}^K c_k \, \e_1^\top \, h_{n}(J_N^k)^\top \, h_{n}(J_N^k)^\top \, \e_1  = \sum_{k=1}^K c_k \,(z_n^k)^\top \, z_k
\end{align*}
and similarly
\begin{align*}
\langle h^*_{n}, h_{n} \rangle_{w} &= \sum_{k=1}^K c_k \, \e_1^\top \, h_{n}(J_N^k)^\top \, (J_N^k)^\top \, h_{n}(J_N^k)^\top \, \e_1  = \sum_{k=1}^K c_k \, (z_n^k)^\top \, J_N^k \, z_k.
\end{align*}

\subsection*{Proof of Proposition~\ref{prop:fnMIX}}

The proof follows several elementary steps
\begin{align*}
f_N &= \int_\R f(x)H_N(x) w(x)dx = \int_\R f(x)H_N(x) \sum_{k=1}^K \, c_k \, v_k(x) dx \\
 &= \sum_{k=1}^K \, c_k \, \int_\R f(x) H_N(x) v_k(x) dx
 = \sum_{k=1}^K  \sum_{n=0}^N \, c_k \, \int_\R q_{N,n}^k  \, f(x) H^k_n(x) v_k(x) dx \\
 &= \sum_{k=1}^K \sum_{n=0}^N \, c_k \, q_{N,n}^k \, \langle f, H_n^k \rangle_{v_k}
\end{align*}
which proves~\eqref{eq:fnMIX} and where the third line results from~\eqref{eq:Hnbasisk} which gives the representation of the polynomial $H_N(x)$ in the polynomials basis $H_n(x)$.

\subsection*{Proof of Proposition~\ref{PROP:lag:coefs}}

Part~\ref{PROP:lag:coefs1}: We want to compute
\[
\e^{rT}f_n = \int_\R (\e^x -e^{\bar k})^+ L_n(x) v_k(x)dx =  \frac{\sqrt{n!}}{\sqrt{\Gamma(\alpha+n)} \Gamma(\alpha)} \int_{\mu}^{\infty} (\e^{\xi + \frac{x}{\beta}} - \e^{\bar k})\Lcal^{\alpha-1}_n(x) \, x^{\alpha-1}\e^{-x}dx
\]
by a change of variable $y=\beta(x-\xi)$, with $\mu=\max(0,\beta({\bar k}-\xi))$ and $v_k(x)$ as in~\eqref{eq:Gammadens}.
We first show that
\[
I_n^{\alpha-1}(\mu;\nu) = \int_\mu^\infty \e^{\nu x} \Lcal^{\alpha-1}_n(x) \, x^{\alpha-1}\e^{-x}dx
\]
satisfies the recursive system~\eqref{eq:Inrecur}.
This directly follows from the recursive relations
\[
\Lcal^{\alpha-1}_n(x) = \left(2+\frac{\alpha - 2}{n} \right) \Lcal^{\alpha-1}_{n-1}(x) - \frac{1}{n} x \, \Lcal^{\alpha-1}_{n-1}(x) - \left(1 + \frac{\alpha - 2}{n} \right)\Lcal^{\alpha-1}_{n-2}(x)
\]
and the three-point rule
\[
\Lcal^{\alpha-1}_n(x) = \Lcal^{\alpha}_n(x) - \Lcal^{\alpha}_{n-1}(x)
\]
such that we obtain
\[
I_n^{\alpha-1}(\mu;\nu) = \left(2+\frac{\alpha - 2}{n} \right)I_{n-1}^{\alpha-1}(\mu;\nu) - \left(1 + \frac{\alpha - 2}{n} \right)I_{n-2}^{\alpha-1}(\mu;\nu) - \frac{1}{n}\left(I_{n-1}^{\alpha}(\mu;\nu) -I_{n-2}^{\alpha}(\mu;\nu) \right).
\]
We conclude by computing
\begin{align*}
I_0^{\alpha-1}(\mu;\nu) &= \int_\mu^\infty x^{\alpha-1} \e^{-(1-\nu)x} dx = \frac{\Gamma(\alpha)}{(1-\nu)^\alpha} \int_\mu^\infty \frac{(1-\nu)^\alpha}{\Gamma(\alpha)} x^{\alpha-1} \e^{-(1-\nu)x} dx \\
&= \frac{\Gamma(\alpha)}{(1-\nu)^\alpha} \left( 1 - \int_0^\mu \frac{(1-\nu)^\alpha}{\Gamma(\alpha)} x^{\alpha-1} \e^{-(1-\nu)x})dx \right) = (1-\nu)^{-\alpha} \left(\Gamma(\alpha) - \Gamma(\alpha,\mu(1-\nu))\right)
\end{align*}
and as $\Lcal_1^{\alpha-1}(x)=(\alpha + x)$ we get that
\[
I_1^{\alpha-1}(\mu;\nu) = \int_\mu^\infty (\alpha + x)\e^{-(1-\nu)x}x^{\alpha-1}dx
= \alpha I_0^{\alpha-1}(\mu;\nu) + I_0^{\alpha}(\mu;\nu)
,\]
which completes the proof of part~\ref{PROP:lag:coefs1}.

Part~\ref{PROP:lag:coefs2}: The Fourier coefficient $f_n$ is given by the expression,
\begin{align*}
\e^{rT}f_n = \int_\R (x-{\bar k})^+ H^k_n(x) v_k(x)dx &=  \int_{\R} (x-{\bar k})^+ H^k_n(x) \frac{\beta^\alpha}{\Gamma(\alpha)} (x-\xi)^{\alpha-1}\e^{-\beta(x-\xi)}dx \\
&= \int_{{\bar k}-\xi}^\infty (y+\xi-{\bar k})H_n(y+\xi) \frac{\beta^\alpha}{\Gamma(\alpha)} y^{\alpha-1}\e^{-\beta y}dy
\end{align*}
where the second equality follows from the change of variable $y=x-\xi$.
We now derive an explicit formula for these integrals which depends on the incomplete Gamma function,
\begin{align*}
\int_{{\bar k}-\xi}^\infty \frac{\beta^\alpha}{\Gamma(\alpha)} y^{j + \alpha-1}\e^{-\beta y}dy &= \frac{1}{\Gamma(\alpha) \beta^j} \int_{\beta({\bar k}-\xi)}^\infty x^{j+\alpha-1}\e^{-x}dx = \frac{\Gamma(\alpha + j, \beta({\bar k}-\xi))}{\Gamma(\alpha) \beta^j}, \quad j\ge 0
\end{align*}
which completes the proof Proposition~\ref{PROP:lag:coefs}.

\subsection*{Proof of Equation~\eqref{eq:gammix_coefs}}

The Gamma density $v_k(x)$ admits an ONB $H_n^k(x)$ given by
\[
H_n^k(x) = \sqrt{\frac{n!}{\Gamma(\alpha+n)}} \Lcal_n^{\alpha-1}(\beta(x-\xi))
\]
where the generalized Laguerre polynomials are recursively given by
\begin{align*}
\Lcal_0^{\alpha-1}(x) &= 1 \\
\Lcal_1^{\alpha-1}(x) &= \alpha + x \\
\Lcal_{n+1}^{\alpha-1}(x) &= \frac{2n + \alpha - x}{n+1} \Lcal_n^{\alpha-1}(x) - \frac{n+\alpha-1}{n+1} \Lcal_{n-1}^{\alpha-1}(x).
\end{align*}
We can rearrange the above expression to get
\[
x \Lcal_n^{\alpha-1}(x) = - (n+1)\Lcal_{n+1}^{\alpha-1}(x) + (2n+\alpha)\Lcal_n^{\alpha-1}(x)  - (n+\alpha-1) \Lcal_{n-1}^{\alpha-1}(x).
\]
Applying the change of variable $x=\beta(y-\xi)$ we get
\[
\beta(y-\xi) \Lcal_n^{\alpha-1}(y) = - (n+1)\Lcal_{n+1}^{\alpha-1}(y) + (2n+\alpha)\Lcal_n^{\alpha-1}(y)  - (n+\alpha-1) \Lcal_{n-1}^{\alpha-1}(y).
\]
which by rearrangement gives
\[
y \Lcal_n^{\alpha-1}(y) = - \frac{n+1}{\beta} \Lcal_{n+1}^{\alpha-1}(y) + \frac{2n+\alpha + \beta\xi}{\beta}\Lcal_n^{\alpha-1}(y)  - \frac{n+\alpha-1}{\beta} \Lcal_{n-1}^{\alpha-1}(y).
\]
Expressing the above in terms of the basis $H_n^k(y)$ gives
\begin{align*}
y H_n(y) &= -\sqrt{\frac{\Gamma(\alpha + n + 1)}{\Gamma(\alpha + n) (n+1)}} \frac{n+1}{\beta} H_{n+1}(y) \\
& \quad + \frac{2n+\alpha + \beta\xi}{\beta}H_n(y)  \\
& \quad - \sqrt{\frac{\Gamma(\alpha + n - 1) n }{\Gamma(\alpha + n)}} \frac{n+\alpha-1}{\beta} H_{n-1}(y).
\end{align*}
Using the relation $\Gamma(\alpha+n+1)/\Gamma(\alpha+n)=\alpha+n$ we get
\begin{align*}
y H_n(y) &= -\frac{\sqrt{(\alpha+n)(n+1)}}{\beta} H_{n+1}(y) \\
& \quad + \frac{2n+\alpha + \beta\xi}{\beta}H_n(y)  \\
& \quad - \frac{\sqrt{(n+\alpha-1)n}}{\beta} H_{n-1}(y).
\end{align*}

\subsection*{Proof of Equations~\eqref{eq:VXcor}--\eqref{eq:volvol}}

The dynamics of the variance $V_t=d\langle X \rangle_t/dt=\Sigma_1(Y_t)^2 + \Sigma_2(Y_t)^2 $ is of the form
\[
dV_t = (\cdots)dt + 2\left( \Sigma'_1(Y_t)\Sigma_1(Y_t) + \Sigma'_2(Y_t)\Sigma_2(Y_t) \right)\, \sigma(Y_t) \, dW_{1t}.
\]
The quadratic covariation between $X_t$ and $V_t$ is therefore given by
\[
d\langle X,V \rangle_t = 2 \left( \Sigma'_1(Y_t)\Sigma_1(Y_t) + \Sigma'_2(Y_t)\Sigma_2(Y_t) \right)  \, \sigma(Y_t) \, \Sigma_1(Y_t) dt
\]
and the quadratic variation of $V_t$ by
\[
d\langle V \rangle_t = 4\left( \Sigma'_1(Y_t)\Sigma_1(Y_t) + \Sigma'_2(Y_t)\Sigma_2(Y_t) \right)^2 \, \sigma(Y_t)^2 dt.
\]
Equation~\eqref{eq:VXcor} directly follows by observing that
\[
\frac{ 2 \left(\Sigma'_1(Y_t)\Sigma_1(Y_t) + \Sigma'_2(Y_t)\Sigma_2(Y_t)\right)\, \sigma(Y_t)}{\sqrt{4 \left( \Sigma'_1(Y_t)\Sigma_1(Y_t) + \Sigma'_2(Y_t)\Sigma_2(Y_t) \right)^2\, \sigma(Y_t)^2}} = \sign\left[\left(\Sigma'_1(Y_t)\Sigma_1(Y_t) + \Sigma'_2(Y_t)\Sigma_2(Y_t)\right) \sigma(Y_t)\right],
\]
and Equation~\eqref{eq:volvol} follows from the above and
\[
d\langle \sqrt{V} \rangle_t = \frac{d \langle V \rangle_t}{4 V_t}.
\]

\subsection*{Proof of Proposition~\ref{prop:condGM}}

Conditional on the trajectory of $W_{1t}$ on $[0,T]$, the trajectory of $Y_t$ is observable and $W_{2t}$ is the only source of randomness in the dynamics of $X_T$ which is thus equivalent to the dynamics of a Gaussian process with time varying parameters.
Hence, its conditional distribution is given by a normal distribution with mean $M_T$ and variance $C_T$ as in~\eqref{eq:MT} and~\eqref{eq:CT}.
Taking expectation gives~\eqref{eq:gGMdens}

\subsection*{Proof of Proposition~\ref{prop:GMparam}}

The first part of the proposition follows from similar arguments as in~\cite[Theorem~3.1]{ackerer2018jacobi}.
Note that it is sufficient to consider only the component of $w$ with the largest variance parameter.

For the second part, it is clear that $C_T$ is bounded when $Y_t$ takes values in a compact interval $I$ because $\Sigma_1 ^2 + \Sigma_2 ^2$ is a polynomial, and thus $\Sigma_1^2$ is bounded on $I$. The random variable $M_T$ is equivalently given by the expression
\[
M_T = X_0 + (r-\delta)T - \frac{1}{2} \int_0^T \left(\Sigma_1(Y_t)^2 + \Sigma_2(Y_t)^2\right)dt + \int_0^T \phi(Y_t) (dY_t - \kappa(\theta - Y_t)dt).
\]
Let $\Phi\in C^2(I)$ be a primitive function of $\phi$, so that $\Phi'=\phi$, on $I$. Applying Ito's lemma to $\Phi(Y_t)$ we obtain
\begin{align*}
 \int_0^T \phi(Y_t) dY_t &= \Phi(Y_T)-\Phi(Y_0) -  \frac{1}{2}  \int_0^T \phi'(Y_t) \sigma(Y_t)^2 dt
\end{align*}
which is uniformly bounded. Hence so is $M_T$.

\subsection*{Proof of Proposition~\ref{prop:deGTv}}

The proof immediately results from the following Lemma which reduces the problem of computing the exponential of a complex matrix into computing the exponential of a real-valued matrix.

\begin{lemma} \label{lem:matexp}
Let $A$ and $B$ be two $n\times n$ real matrices, then
\begin{equation}\label{eq:BSSmatexp}
\exp\begin{pmatrix}
A & -B\\
B & A
\end{pmatrix}
=
\begin{pmatrix}
R & -S \\
S & R
\end{pmatrix}
\end{equation}
for some $n\times n$ real matrices $R$ and $S$.
Furthermore we have $\exp(A + \im B) = R + \im S$.
\end{lemma}

\begin{proof}[Proof of Lemma~\ref{lem:matexp}]
First observe that for any $n \times n$ matrices $C$ and $D$ we have
\[
\begin{pmatrix}
A & -B\\
B & A
\end{pmatrix}
\begin{pmatrix}
C & -D\\
D & C
\end{pmatrix}
=\begin{pmatrix}
AC - BD & -AD - BC\\
AD + BC & AC - BD
\end{pmatrix}
\]
so that the block skew-symmetric structure is preserved.
The Equation~\eqref{eq:BSSmatexp} then follows from the  power series representation of the matrix exponential.
Consider the deterministic matrix-valued process $Y_t$, whose dynamics is given by
\[
dY_t =
\begin{pmatrix}
A & -B\\
B & A
\end{pmatrix}
Y_t \, dt, \quad
Y_0 =
\begin{pmatrix}
I_n \\
0_n
\end{pmatrix},
\]
and the matrix-valued process
$X_t= \begin{pmatrix} I_n & \im I_n \end{pmatrix} Y_t$
so that
$dX_t = (A + \im B)X_t \, dt$ and $X_0 = I_n$.
Then we have
\[
Y_1 = \begin{pmatrix}
R & -S \\
S & R
\end{pmatrix}
\begin{pmatrix}
I_n \\
0_n
\end{pmatrix} =
\begin{pmatrix}
R \\
S
\end{pmatrix}
\]
and, from the definitions of $X_t$ and $Y_t$, it follows that
\[
X_1 = \exp(A + \im B) =
\begin{pmatrix}
I_n & \im I_n
\end{pmatrix} Y_1
=R + \im S.
\]
\end{proof}

\section{Basis Construction with Moments} \label{sec:ONBmts}

In this Appendix we present moment-based constructions, alternative to Proposition~\ref{prop:ONBmix}, for the orthonormal basis (ONB) $H_n(x)$ of the space $L_w^2$ which can also be used when the auxiliary density $w$ is $d$-valued.
Let $\pi:\Ecal\rightarrow\{1,\dots,M\}$ be an enumeration of the set of exponents
\[
\Ecal=\left\{ n\in\N^d : |n| \le N \right\}
\]
for some positive integer $N$, with $\pi(0)=1$, and such that $\pi(n) \le \pi(m)$ if $|n|\le|m|$.
We denote $\pi_i=\pi^{-1}(i)\in\N^d$ where $\pi^{-1}:\{1,\dots,M\}\rightarrow \Ecal$ is the inverse function of $\pi$.

A standard approach to construct the ONB is to apply the Gram-Schmidt algorithm outlined below.
First one constructs the orthogonal basis
\begin{align*}
u_0(x) & = 1 \\
u_i(x) &= x^{\pi_i} - \sum_{j=0}^{i-1} \frac{\langle x^{\pi_i}, u_j \rangle_w}{\langle u_j, u_j \rangle_w} \, u_j(x), \quad i\ge 1
\end{align*}
and the ONB is obtained by normalization,
\[
H_i(x) = \frac{u_i(x)}{\lVert u_i \rVert_w}, \quad i\ge 0.
\]

Another interesting approach to construct an ONB of the space $L_w^2$ is as follows.
Let $\bf M$ denote the $(M\times M)$ Gram matrix defined by
\begin{equation}\label{eq:Grammat}
{\bf M}_{i+1,\,j+1} = \langle x^{\pi_i}, x^{\pi_j} \rangle_w
\end{equation}
which is thus symmetric and positive definite.
Let ${\bf M= LL^\top}$ be the unique Cholesky decomposition of ${\bf M}$ where ${\bf L}$ is a lower triangular matrix, and defined the lower triangular ${\bf S} = {\bf L}^{-1}$.

\begin{theorem}[\cite{mysovskikh1968construction}]
The polynomials
\[
H_{i}(x) = \sum_{j=0}^i {\bf S}_{i+1,j+1} \; x^{\pi_j}
\]
form an ONB of $L_w^2$.
\end{theorem}

\begin{remark}
The orthonormal basis resulting from the classical Gram-Schmidt implementation described above may not appear orthogonal numerically because of rounding errors, so the procedure is said to be numerically unstable.
To alleviate this issue the modified Gram-Schmidt implementation is often preferred in practice, the polynomial $u_i(x)$ is now computed in multiple steps
\[
u_i^{(j+1)}(x) = u_i^{(j)}(x) - \frac{\langle u_i^{(j)}(x), u_j \rangle_w}{\langle u_j, u_j \rangle_w} \, u_j(x), \quad j=0,\dots,i-1
\]
with $u_i^{(0)}(x) = x^{\pi_i}$ and such that $u_i(x)=u_i^{(i)}(x)$.
Although the two algorithms are equivalent in exact arithmetic, significant difference can be observed in finite-precision arithmetic.
\end{remark}

\begin{remark}
The Gram matrix in Equation~\eqref{eq:Grammat} may numerically be singular because of rounding errors either in the computation in its eigenvalues or its moments.
One approach to avoid this problem is to consider an approximately orthonormal basis in place of the monomial basis.
By doing so, the Gram matrix would be already almost diagonal and thus more likely to be invertible.
This may be achieved, for example, by implementing an algorithm that computes the ONB for the first $j$ elements by using the ONB of the first $j-1$ elements enlarged with the monomial $x^{\pi_j}$.
\end{remark}

\section{Heston Option Price and Greeks} \label{sec:Heston}

The option price formula for the Heston model can be found in many textbooks, in this section we use the representation found in~\cite{gatheral2011volatility}.
With the notations of Section~\ref{sec:calib}, the option price for a given log strike ${\bar k}$ is given by
\[
\pi = \e^{(r-\delta)T + x_0}P_1 - \e^{\bar k} P_0
\]
where the pseudo probabilities $P_0$ and $P_1$ are given by
\[
P_j = \frac{1}{2} + \frac{1}{\pi} \int_0^\infty \Re\left(\frac{\exp\left( C_j \theta + D_jv_0 + \im u ((r-\delta)T +x_0 - {\bar k}) \right)}{\im u} \right) du, \quad \text{for $j=0,1$}
\]
with the following elements which depend on $u$ and $j$
\begin{center}
\begin{tabular}{L  L }
C_j = \kappa \left(c_- T - \frac{2}{\sigma^2} \log \left( \frac{1 - (c_-/c_+)\e^{-c_0 T} }{1- c_-/c_+}\right)\right) &
D_j = c_- \frac{1 - \e^{-c_0 T}}{1 - (c_-/c_+)\e^{-c_0 T}}\\
c_0 = \sqrt{\beta^2- 2\alpha\sigma^2} & \alpha = -\frac{u(u + \im)}{2} + \im j u \\
c_\pm = (\beta \pm c_0)/\sigma^2 & \beta = \kappa - \rho \sigma (j + \im u)
\end{tabular}
\end{center}
The option Delta and Gamma are given by
\begin{align*}
\Delta_f &= \e^{(r-\delta)T}P_1 + \e^{(r-\delta)T+ x_0 } \frac{\partial P_1}{\partial \e^{x_0}} - \e^{\bar k} \frac{\partial P_0}{\partial \e^{x_0}}\\
\Gamma_f &= 2 \e^{(r-\delta)T} \frac{\partial P_1}{\partial \e^{x_0}} + \e^{(r-\delta)T+ x_0 } \frac{\partial^2 P_1}{\partial \e^{x_0} \partial \e^{x_0}} - \e^{\bar k} \frac{\partial^2 P_0}{\partial \e^{x_0} \partial \e^{x_0}}
\end{align*}
with the partial derivatives
\begin{align*}
\frac{\partial P_j}{\partial \e^{x_0}} &=  \frac{1}{\pi} \int_0^\infty \Re\left(\exp\left( C_j \theta + D_jv_0 + \im u (x_0 -1) + \im u ((r-\delta)T - {\bar k}) \right)\right) du \\
\frac{\partial^2 P_j}{\partial \e^{x_0} \partial \e^{x_0}} &=  \frac{1}{\pi} \int_0^\infty \Re\left((\im u-1)\exp\left( C_j \theta + D_jv_0 + \im u (x_0 -2) + \im u ((r-\delta)T +x_0 - {\bar k}) \right) \right) du
\end{align*}
for $j=0,1$.
\end{appendix}


\begin{table}
\begin{center}
\begin{tabular}{l||rr|rr|rr}
& \multicolumn{2}{c|}{${\bar k}=-0.1$} & \multicolumn{2}{c|}{${\bar k}=0$} & \multicolumn{2}{c}{${\bar k}=0.1$}\\
$N$ & $K=1$ & $K=2$ &  $K=1$ & $K=2$ & $K=1$ & $K=2$ \\
  \hline  \hline
0--1 & 23.59 & 3.67 & 26.67 & 1.02 & 25.97 & 3.25 \\
  2 & 3.77 & 3.67 & 7.47 & 1.02 & 2.88 & 3.25 \\
  3 & 2.63 & 1.89 & 7.36 & 0.87 & - & 0.17 \\
  4 & 7.38 & 1.86 & 4.03 & 0.77 & - & 0.03 \\
  5 & 5.55 & 1.01 & 4.00 & 0.72 & - & 2.47 \\
  6 & 5.18 & 0.88 & 2.57 & 0.58 & - & 1.80 \\
  7 & 4.22 & 0.55 & 2.55 & 0.55 & - & 3.00 \\
  8 & 3.23 & 0.38 & 1.78 & 0.42 & - & 2.10 \\
  9 & 2.73 & 0.28 & 1.77 & 0.40 & - & 2.32 \\
  10 & 1.91 & 0.12 & 1.30 & 0.29 & - & 1.66 \\
  11 & 1.64 & 0.12 & 1.29 & 0.28 & - & 1.57 \\
  12 & 1.02 & 0.00 & 0.98 & 0.19 & - & 1.15 \\
  13 & 0.88 & 0.02 & 0.97 & 0.19 & - & 1.02 \\
  14 & 0.42 & 0.06 & 0.76 & 0.12 & - & 0.78 \\
  15 & 0.37 & 0.04 & 0.75 & 0.12 & - & 0.67 \\
  16 & 0.04 & 0.09 & 0.60 & 0.07 & - & 0.53 \\
  17 & 0.03 & 0.06 & 0.59 & 0.07 & - & 0.44 \\
  18 & 0.21 & 0.09 & 0.48 & 0.04 & 5.12 & 0.38 \\
  19 & 0.18 & 0.07 & 0.47 & 0.03 & 4.43 & 0.31 \\
  20 & 0.35 & 0.08 & 0.39 & 0.01 & 3.12 & 0.28 \\
  30 & 0.39 & 0.00 & 0.15 & 0.01 & 0.49 & 0.04 \\
  40 & 0.15 & 0.04 & 0.06 & 0.01 & 0.16 & 0.09 \\
  50 & 0.02 & 0.04 & 0.02 & 0.01 & 0.32 & 0.10 \\
\end{tabular}
\end{center}
\caption{Implied volatility errors for the Jacobi model.}
\scriptsize
The reported values are absolute percentage errors with respect to implied volatility approximations obtained at the 100-th truncation order for call options with different log strikes ${\bar k}$.
The auxiliary density is a Gaussian mixture with two components  whose two first moments match those of the log price.
The ''--'' symbol indicates that the implied volatility was not retrievable because the option price approximation was negative.
\label{tab:IV-GM}
\end{table}


\begin{table}
\begin{center}
\begin{tabular}{l||rr|rr|rr}
& \multicolumn{2}{c|}{$K=3$} & \multicolumn{2}{c|}{$K=10$} & \multicolumn{2}{c}{$K=50$}\\
$N$ & GM & GM$^+$ &  GM & GM$^+$ &  GM & GM$^+$   \\
  \hline  \hline
0--1 & 1.26 & 0.64 & 0.01 & 1.85 & 0.16 & 1.95 \\
  2 & 0.75 & 0.12 & 0.20 & 0.24 & 0.10 & 0.24 \\
  3 & 0.26 & 0.02 & 0.15 & 0.19 & 0.16 & 0.18 \\
  4 & 0.43 & 0.07 & 0.09 & 0.05 & 0.15 & 0.05 \\
  5 & 0.08 & 0.03 & 0.06 & 0.03 & 0.08 & 0.03 \\
  6 & 0.68 & 0.06 & 0.04 & 0.01 & 0.03 & 0.02 \\
  7 & 0.46 & 0.03 & 0.01 & 0.00 & 0.04 & 0.01 \\
  8 & 1.27 & 0.04 & 0.01 & 0.01 & 0.01 & 0.02 \\
  9 & 1.56 & 0.03 & 0.03 & 0.01 & 0.00 & 0.01 \\
  10 & 2.94 & 0.03 & 0.04 & 0.02 & 0.01 & 0.03 \\
  11 & 6.06 & 0.02 & 0.06 & 0.01 & 0.01 & 0.02 \\
  12 & 8.02 & 0.02 & 0.04 & 0.02 & 0.01 & 0.03 \\
  13 & 25.63 & 0.01 & 0.05 & 0.02 & 0.01 & 0.02 \\
  14 & 23.63 & 0.01 & 0.18 & 0.02 & 0.02 & 0.03 \\
  15 & -- & 0.01 & 0.10 & 0.02 & 0.01 & 0.03 \\
  16 & -- & 0.01 & 0.43 & 0.02 & 0.02 & 0.03 \\
  17 & -- & 0.01 & 0.03 & 0.02 & 0.01 & 0.03 \\
  18 & -- & 0.01 & 1.30 & 0.02 & 0.02 & 0.02 \\
  19 & -- & 0.01 & 0.63 & 0.02 & 0.02 & 0.02 \\
  20 & -- & 0.00 & 3.82 & 0.02 & 0.02 & 0.02 \\
  \hline
  30 & -- & 0.00 & -- & 0.00 & 0.00 & 0.00 \\
  40 & -- & 0.00 & -- & 0.00 & 4.65 & 0.00 \\
  50 & -- & 4.34 & -- & 4.60 & -- & 9.96 \\
\end{tabular}
\end{center}
\caption{Implied volatility errors for the Stein--Stein model.}
\scriptsize
The reported values are absolute percentage errors with respect to the implied volatility computed with Fourier transform technique.
The GM$^+$ column refers to option price approximations obtained with a Gaussian mixture auxiliary density whose $20$-th moment is matching $\E[X_T^{20}]$.
The ''--'' symbol indicates that the implied volatility was not retrievable either because the option price approximation was negative or because the implied volatility was larger than 99\%.
\label{tab:SteinIV}
\end{table}

\begin{table}
\begin{center}
\begin{tabular}{l||rr|rr|rr}
& \multicolumn{2}{c|}{$K=3$} & \multicolumn{2}{c|}{$K=10$} & \multicolumn{2}{c}{$K=50$}\\
$N$ & GM & GM$^+$ &  GM & GM$^+$ &  GM & GM$^+$   \\
  \hline  \hline
  0--1 & 19.26 & 20.98 & 20.32 & 21.98 & 20.41 & 22.06 \\
  2 & 20.91 & 20.34 & 20.47 & 20.46 & 20.37 & 20.44 \\
  3 & 20.52 & 20.27 & 20.41 & 20.42 & 20.41 & 20.41 \\
  4 & 19.97 & 20.32 & 20.36 & 20.31 & 20.40 & 20.29 \\
  5 & 20.36 & 20.29 & 20.35 & 20.29 & 20.36 & 20.28 \\
  6 & 20.83 & 20.33 & 20.34 & 20.29 & 20.33 & 20.29 \\
  7 & 19.99 & 20.31 & 20.30 & 20.28 & 20.34 & 20.28 \\
  8 & 19.36 & 20.33 & 20.32 & 20.30 & 20.33 & 20.30 \\
  9 & 21.61 & 20.32 & 20.38 & 20.30 & 20.32 & 20.30 \\
  10 & 22.71 & 20.33 & 20.30 & 20.32 & 20.33 & 20.32 \\
  11 & 14.88 & 20.32 & 20.25 & 20.31 & 20.33 & 20.31 \\
  12 & 13.37 & 20.33 & 20.44 & 20.33 & 20.33 & 20.33 \\
  13 & 47.02 & 20.33 & 20.48 & 20.33 & 20.33 & 20.33 \\
  14 & 42.79 & 20.33 & 19.93 & 20.33 & 20.33 & 20.34 \\
  15 & -- & 20.33 & 20.10 & 20.33 & 20.33 & 20.34 \\
  16 & -- & 20.33 & 21.94 & 20.34 & 20.33 & 20.34 \\
  17 & -- & 20.33 & 20.29 & 20.34 & 20.33 & 20.34 \\
  18 & 66.98 & 20.33 & 13.62 & 20.34 & 20.34 & 20.34 \\
  19 & -- & 20.33 & 24.63 & 20.34 & 20.34 & 20.34 \\
  20 & -- & 20.33 & 50.84 & 20.34 & 20.32 & 20.34 \\
    \hline
  30 & -- & 20.39 & -- & 20.40 & 37.42 & 20.41 \\
  40 & -- & -- & -- & -- & -- & -- \\
  50 & -- & -- & -- & -- & -- & -- \\
\end{tabular}
\end{center}
\caption{Implied volatility for the Hull--White model.}
\scriptsize
The reported values are in percentage.
The GM$^+$ column refers to option price approximations obtained with a Gaussian mixture auxiliary density whose $20$-th moment is matching $\E[X_T^{20}]$.
The ''--'' symbol indicates that the implied volatility was not retrievable either because the option price approximation was negative or because the implied volatility was larger than 99\%.
\label{tab:HullWhiteIV}
\end{table}


\begin{figure}
\begin{center}
\includegraphics{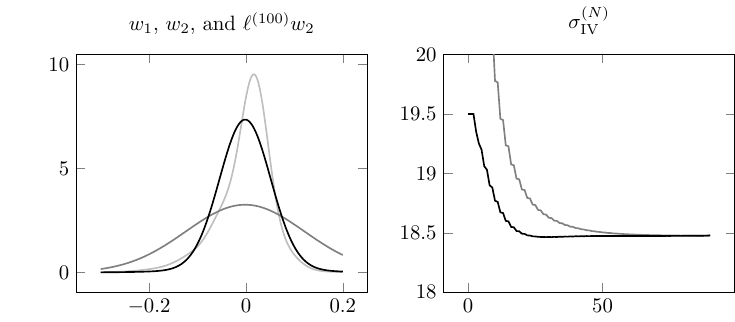}
\end{center}
\caption{Auxiliary densities and implied volatility convergence.}
\scriptsize
The left panel displays the Gausian mixture used as auxiliary density with one (grey) and two (black) components as well as the log price density approximation at the order $N=100$ (light gray).
The right panel displays the implied volatility series for the corresponding two auxiliary densities for a call option with strike ${\bar k}=0$.
\label{fig:IV-GM}
\end{figure}


\begin{figure}
\begin{center}
\includegraphics{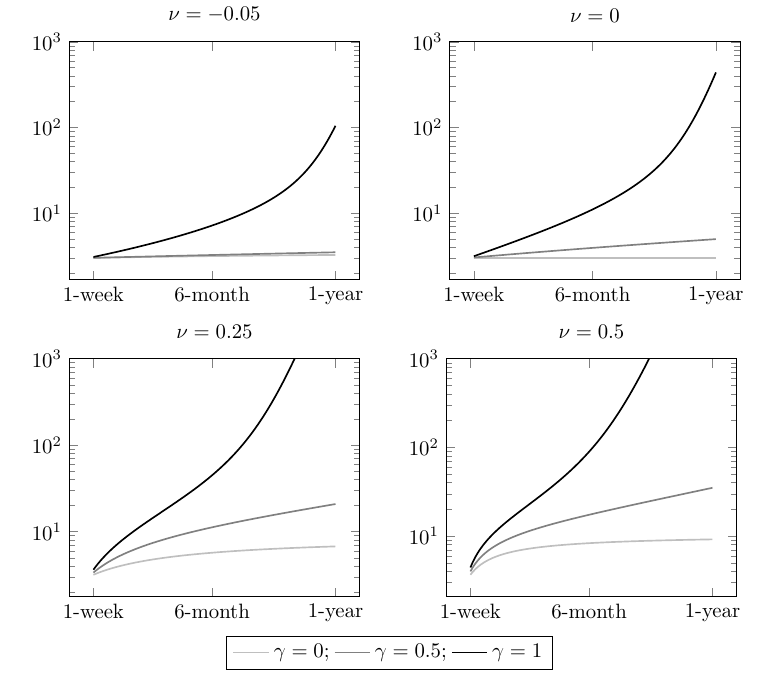}
\end{center}
\caption{Kurtosis values for the Hull--White model.}
\scriptsize
The kurtosis is displayed for maturities between one week and on year, and for different combination of the $\nu$ and $\gamma$ parameters.
The light gray line $\gamma=0$ gives the kurtosis of the nested Stein--Stein model.
\label{fig:HWkurtosis}
\end{figure}

\begin{figure}
\begin{center}
\includegraphics{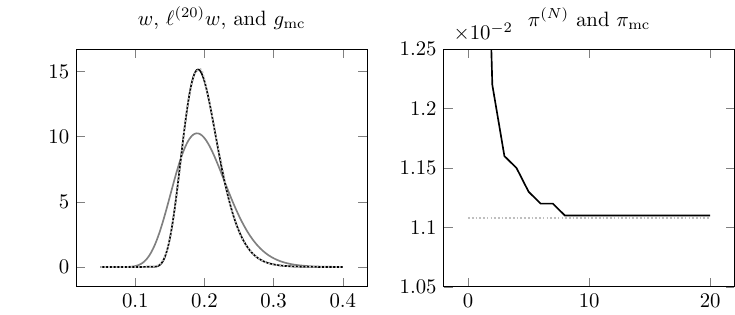}
\end{center}
\caption{Realized variance density and option price approximation for the GARCH model.}
\scriptsize
The left panel displays the auxiliary density (gray), the density approximation at the order $N=20$ (black), and the simulated density (dotted light gray) from 100'000 paths.
Note that the last two densities almost perfectly coincide.
The right panel displays the option price series for a call option with strike ${\bar k}=0.20$ (black) as well as the Monte-Carlo price (dotted light gray).
\label{fig:RVoption}
\end{figure}

\begin{figure}
\begin{center}
\includegraphics{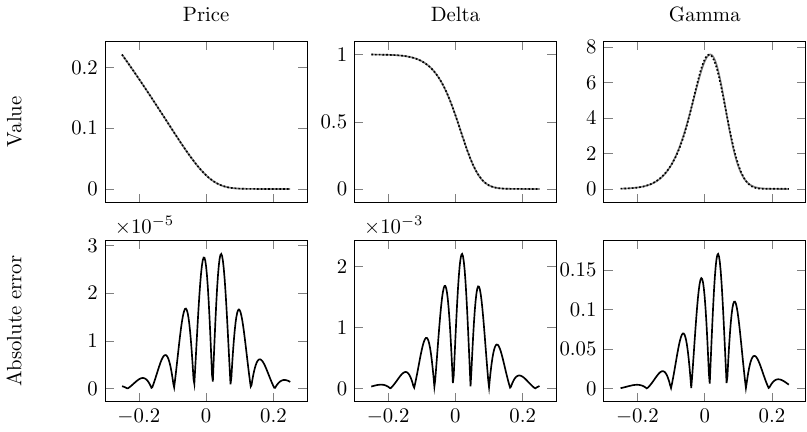}
\end{center}
\caption{Option price, Delta, and Gamma for the Heston model.}
\scriptsize
Panels on the first row display the price (left), the Delta (middle), and the Gamma (right) for the Heston computed with the Fourier transform formula (gray) and the orthogonal polynomial expansion (black dotted).
The second row displays the absolute errors between the Fourier transform and approximation values.
\label{fig:greeks}
\end{figure}

\begin{figure}
\begin{center}
\includegraphics{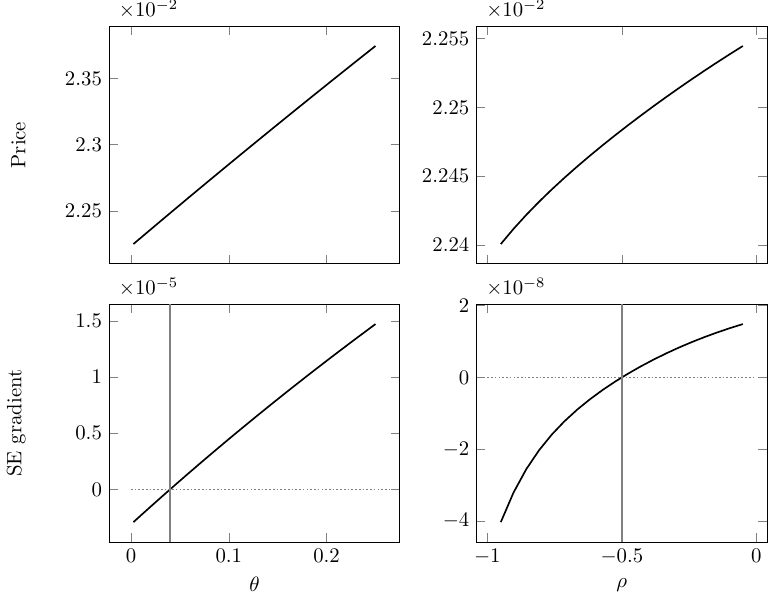}
\end{center}
\caption{Theta and Rho sensitivities and gradients for the Heston model.}
\scriptsize
The first row displays the approximated option price when the $\theta$ (left) and the $\rho$ (right) parameters are varied and using the auxiliary density of the reference parameters.
The second row displays the gradient of the squared price error with respect to the reference price whose parameters value are indicated by the vertical lines.
The horizontal dotted line emphasizes that the zero gradient is attained for the reference parameter values and that the gradient's sign is correct.
\label{fig:parsens}
\end{figure}

\begin{figure}
\begin{center}
\includegraphics{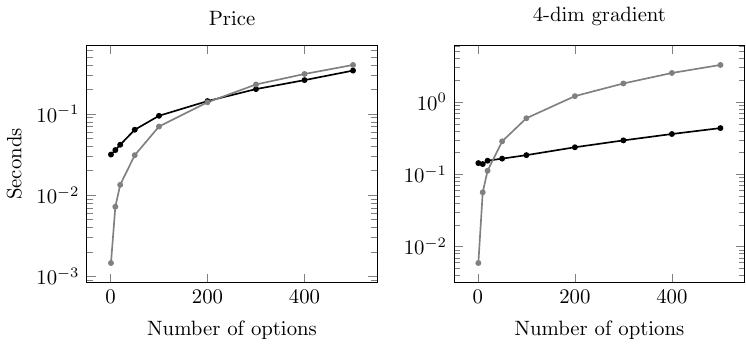}
\end{center}
\caption{Option price and gradient CPU time.}
\scriptsize
The left panel displays the CPU time required to compute option prices for the Heston model using the Fourier transform formula (gray) and the orthogonal polynomial expansion (black) for different number of contracts.
The right panel displays the CPU time required to compute the option price gradient with respect to $\kappa$, $\theta$, $\sigma$, and $\rho$ for different number of contracts.
\label{fig:cpu}
\end{figure}

\begin{figure}
\begin{center}
\includegraphics{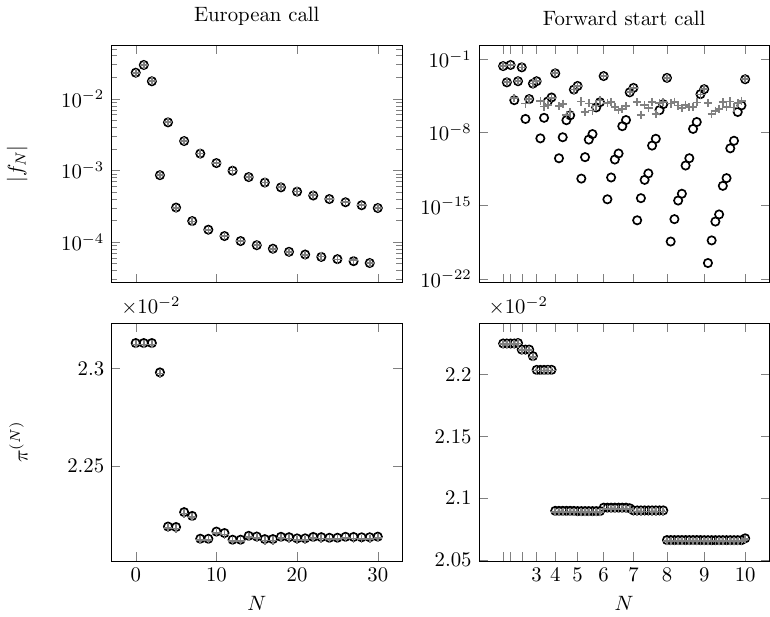}
\end{center}
\caption{Weighted least squares approximation of the payoff coefficients.}
\scriptsize
The first row displays the exact (circles) and approximated (crosses) payoff coefficients for the call and forward start call options.
The second row displays the corresponding option price series for the Jacobi model.
\label{fig:payoffapp}
\end{figure}

\processdelayedfloats

\newpage
\bibliographystyle{chicago}
\bibliography{OPE}

\begin{thebibliography}{}

\bibitem[\protect\citeauthoryear{Ackerer and Filipovi{\'c}}{Ackerer and
  Filipovi{\'c}}{2016}]{ackerer2016linear}
Ackerer, D. and D.~Filipovi{\'c} (2016).
\newblock Linear credit risk models.
\newblock {\em Swiss Finance Institute Research Paper\/}~(16-34).

\bibitem[\protect\citeauthoryear{Ackerer, Filipovi{\'c}, and Pulido}{Ackerer
  et~al.}{2018}]{ackerer2018jacobi}
Ackerer, D., D.~Filipovi{\'c}, and S.~Pulido (2018).
\newblock The {J}acobi stochastic volatility model.
\newblock {\em Finance and Stochastics\/}~{\em 22\/}(3), 667--700.

\bibitem[\protect\citeauthoryear{Ackerer and Willems}{Ackerer and
  Willems}{2019}]{ackerer2018poly}
Ackerer, D. and S.~Willems (2019).
\newblock A polynomial approximation of stochastic volatility models.
\newblock {\em Working Paper\/}.

\bibitem[\protect\citeauthoryear{Al-Mohy and Higham}{Al-Mohy and
  Higham}{2010}]{al2010complex}
Al-Mohy, A.~H. and N.~J. Higham (2010).
\newblock The complex step approximation to the {F}r{\'e}chet derivative of a
  matrix function.
\newblock {\em Numerical Algorithms\/}~{\em 53\/}(1), 133.

\bibitem[\protect\citeauthoryear{B{\"a}ck, Nobile, Tamellini, and
  Tempone}{B{\"a}ck et~al.}{2011}]{back2011stochastic}
B{\"a}ck, J., F.~Nobile, L.~Tamellini, and R.~Tempone (2011).
\newblock Stochastic spectral {G}alerkin and collocation methods for {PDE}s
  with random coefficients: {A} numerical comparison.
\newblock In {\em Spectral and high order methods for partial differential
  equations}, pp.\  43--62. Springer.

\bibitem[\protect\citeauthoryear{Barone-Adesi, Rasmussen, and
  Ravanelli}{Barone-Adesi et~al.}{2005}]{barone2005option}
Barone-Adesi, G., H.~Rasmussen, and C.~Ravanelli (2005).
\newblock An option pricing formula for the {GARCH} diffusion model.
\newblock {\em Computational Statistics \& Data Analysis\/}~{\em 49\/}(2),
  287--310.

\bibitem[\protect\citeauthoryear{Beck, Tempone, Nobile, and Tamellini}{Beck
  et~al.}{2012}]{beck2012optimal}
Beck, J., R.~Tempone, F.~Nobile, and L.~Tamellini (2012).
\newblock On the optimal polynomial approximation of stochastic {PDE}s by
  {G}alerkin and collocation methods.
\newblock {\em Mathematical Models and Methods in Applied Sciences\/}~{\em
  22\/}(09), 1250023.

\bibitem[\protect\citeauthoryear{Black and Scholes}{Black and
  Scholes}{1973}]{black1973pricing}
Black, F. and M.~Scholes (1973).
\newblock The pricing of options and corporate liabilities.
\newblock {\em Journal of political economy\/}~{\em 81\/}(3), 637--654.

\bibitem[\protect\citeauthoryear{Broadie and Kaya}{Broadie and
  Kaya}{2004}]{broadie2004exact}
Broadie, M. and O.~Kaya (2004).
\newblock Exact simulation of option {G}reeks under stochastic volatility and
  jump diffusion models.
\newblock In {\em Simulation Conference, 2004. Proceedings of the 2004 Winter},
  Volume~2, pp.\  1607--1615. IEEE.

\bibitem[\protect\citeauthoryear{Callegaro, Fiorin, and Grasselli}{Callegaro
  et~al.}{2017a}]{callegaro2017pricing}
Callegaro, G., L.~Fiorin, and M.~Grasselli (2017a).
\newblock Pricing via recursive quantization in stochastic volatility models.
\newblock {\em Quantitative Finance\/}~{\em 17\/}(6), 855--872.

\bibitem[\protect\citeauthoryear{Callegaro, Fiorin, and Grasselli}{Callegaro
  et~al.}{2017b}]{callegaro2017quantization}
Callegaro, G., L.~Fiorin, and M.~Grasselli (2017b).
\newblock Quantization meets {F}ourier: A new technology for pricing options.
\newblock {\em Working Paper\/}.

\bibitem[\protect\citeauthoryear{Callegaro, Fiorin, and Pallavicini}{Callegaro
  et~al.}{2017}]{callegaro2017quantizationb}
Callegaro, G., L.~Fiorin, and A.~Pallavicini (2017).
\newblock Quantization goes polynomial.
\newblock {\em Working Paper\/}.

\bibitem[\protect\citeauthoryear{Carr and Madan}{Carr and
  Madan}{1999}]{carr1999option}
Carr, P. and D.~Madan (1999).
\newblock Option valuation using the fast {F}ourier transform.
\newblock {\em Journal of Computational Finance\/}~{\em 2\/}(4), 61--73.

\bibitem[\protect\citeauthoryear{Carr and Schoutens}{Carr and
  Schoutens}{2008}]{carr2008hedging}
Carr, P. and W.~Schoutens (2008).
\newblock Hedging under the {H}eston model with jump-to-default.
\newblock {\em International Journal of Theoretical and Applied Finance\/}~{\em
  11\/}(04), 403--414.

\bibitem[\protect\citeauthoryear{Chan, Joshi, and Zhu}{Chan
  et~al.}{2015}]{chan2015first}
Chan, J.~H., M.~S. Joshi, and D.~Zhu (2015).
\newblock First-and second-order {G}reeks in the {H}eston model.

\bibitem[\protect\citeauthoryear{Cohen, Davenport, and Leviatan}{Cohen
  et~al.}{2013}]{cohen2013stability}
Cohen, A., M.~A. Davenport, and D.~Leviatan (2013).
\newblock On the stability and accuracy of least squares approximations.
\newblock {\em Foundations of computational mathematics\/}~{\em 13\/}(5),
  819--834.

\bibitem[\protect\citeauthoryear{Cohen, Devore, and Schwab}{Cohen
  et~al.}{2011}]{cohen2011analytic}
Cohen, A., R.~Devore, and C.~Schwab (2011).
\newblock Analytic regularity and polynomial approximation of parametric and
  stochastic elliptic {PDE}'s.
\newblock {\em Analysis and Applications\/}~{\em 9\/}(01), 11--47.

\bibitem[\protect\citeauthoryear{Cohen and Migliorati}{Cohen and
  Migliorati}{2017}]{cohen2017optimal}
Cohen, A. and G.~Migliorati (2017).
\newblock Optimal weighted least-squares methods.
\newblock {\em SMAI Journal of Computational Mathematics\/}~{\em 3}, 181--203.

\bibitem[\protect\citeauthoryear{Davydov and Linetsky}{Davydov and
  Linetsky}{2003}]{davydov2003pricing}
Davydov, D. and V.~Linetsky (2003).
\newblock Pricing options on scalar diffusions: {A}n eigenfunction expansion
  approach.
\newblock {\em Operations research\/}~{\em 51\/}(2), 185--209.

\bibitem[\protect\citeauthoryear{Drost and Werker}{Drost and
  Werker}{1996}]{drost1996closing}
Drost, F.~C. and B.~J. Werker (1996).
\newblock Closing the {GARCH} gap: Continuous time {GARCH} modeling.
\newblock {\em Journal of Econometrics\/}~{\em 74\/}(1), 31--57.

\bibitem[\protect\citeauthoryear{Duffie, Filipovi\'{c}, and
  Schachermayer}{Duffie et~al.}{2003}]{duffie2003affine}
Duffie, D., D.~Filipovi\'{c}, and W.~Schachermayer (2003).
\newblock Affine processes and applications in finance.
\newblock {\em Annals of Applied Probability\/}~{\em 13\/}(3), 984--1053.

\bibitem[\protect\citeauthoryear{Dunkl and Xu}{Dunkl and
  Xu}{2014}]{dunkl2014orthogonal}
Dunkl, C.~F. and Y.~Xu (2014).
\newblock {\em Orthogonal polynomials of several variables}.
\newblock Cambridge University Press.
\newblock Second Edition.

\bibitem[\protect\citeauthoryear{Fang and Oosterlee}{Fang and
  Oosterlee}{2009}]{fang2009novel}
Fang, F. and C.~W. Oosterlee (2009).
\newblock A novel pricing method for {E}uropean options based on
  {F}ourier-cosine series expansions.
\newblock {\em {SIAM} Journal on Scientific Computing\/}~{\em 31\/}(2),
  826--848.

\bibitem[\protect\citeauthoryear{Feng, Kovalov, Linetsky, and Marcozzi}{Feng
  et~al.}{2007}]{feng2007variational}
Feng, L., P.~Kovalov, V.~Linetsky, and M.~Marcozzi (2007).
\newblock Variational methods in derivatives pricing.
\newblock {\em Handbooks in Operations Research and Management Science\/}~{\em
  15}, 301--342.

\bibitem[\protect\citeauthoryear{Filipovi{\'c}, Gourier, and
  Mancini}{Filipovi{\'c} et~al.}{2016}]{filipovic2016quadratic}
Filipovi{\'c}, D., E.~Gourier, and L.~Mancini (2016).
\newblock Quadratic variance swap models.
\newblock {\em Journal of Financial Economics\/}~{\em 119\/}(1), 44--68.

\bibitem[\protect\citeauthoryear{Filipovi\'c and Larsson}{Filipovi\'c and
  Larsson}{2017}]{filipovic2017polynomial}
Filipovi\'c, D. and M.~Larsson (2017).
\newblock Polynomial jump-diffusion models.
\newblock {\em Working Paper\/}.

\bibitem[\protect\citeauthoryear{Filipovi{\'c}, Mayerhofer, and
  Schneider}{Filipovi{\'c} et~al.}{2013}]{filipovic2013density}
Filipovi{\'c}, D., E.~Mayerhofer, and P.~Schneider (2013).
\newblock Density approximations for multivariate affine jump-diffusion
  processes.
\newblock {\em Journal of Econometrics\/}~{\em 176\/}(2), 93--111.

\bibitem[\protect\citeauthoryear{Fischer and Golub}{Fischer and
  Golub}{1992}]{fischer1992generate}
Fischer, B. and G.~H. Golub (1992).
\newblock How to generate unknown orthogonal polynomials out of known
  orthogonal polynomials.
\newblock {\em Journal of Computational and Applied Mathematics\/}~{\em
  43\/}(1-2), 99--115.

\bibitem[\protect\citeauthoryear{Gatheral}{Gatheral}{2011}]{gatheral2011volatility}
Gatheral, J. (2011).
\newblock {\em The volatility surface: {A} practitioner's guide}.
\newblock John Wiley \& Sons.

\bibitem[\protect\citeauthoryear{Gautschi}{Gautschi}{2004}]{gau_04}
Gautschi, W. (2004).
\newblock {\em Orthogonal Polynomials: Computation and Approximation}.
\newblock Numerical Mathematics and Scientific Computation. Oxford University
  Press, New York.
\newblock Oxford Science Publications.

\bibitem[\protect\citeauthoryear{Glasserman and Kim}{Glasserman and
  Kim}{2011}]{glasserman2011gamma}
Glasserman, P. and K.-K. Kim (2011).
\newblock Gamma expansion of the {H}eston stochastic volatility model.
\newblock {\em Finance and Stochastics\/}~{\em 15\/}(2), 267--296.

\bibitem[\protect\citeauthoryear{Hampton and Doostan}{Hampton and
  Doostan}{2015}]{hampton2015compressive}
Hampton, J. and A.~Doostan (2015).
\newblock Compressive sampling of polynomial chaos expansions: {C}onvergence
  analysis and sampling strategies.
\newblock {\em Journal of Computational Physics\/}~{\em 280}, 363--386.

\bibitem[\protect\citeauthoryear{Heston}{Heston}{1993}]{heston1993closed}
Heston, S.~L. (1993).
\newblock A closed-form solution for options with stochastic volatility with
  applications to bond and currency options.
\newblock {\em The Review of Financial Studies\/}~{\em 6\/}(2), 327--343.

\bibitem[\protect\citeauthoryear{Heston and Rossi}{Heston and
  Rossi}{2016}]{heston2016spanning}
Heston, S.~L. and A.~G. Rossi (2016).
\newblock A spanning series approach to options.
\newblock {\em The Review of Asset Pricing Studies\/}~{\em 7\/}(1), 2--42.

\bibitem[\protect\citeauthoryear{Hull and White}{Hull and
  White}{1987}]{hull1987pricing}
Hull, J. and A.~White (1987).
\newblock The pricing of options on assets with stochastic volatilities.
\newblock {\em The Journal of Finance\/}~{\em 42\/}(2), 281--300.

\bibitem[\protect\citeauthoryear{Kahl and J{\"a}ckel}{Kahl and
  J{\"a}ckel}{2006}]{kahl2006fast}
Kahl, C. and P.~J{\"a}ckel (2006).
\newblock Fast strong approximation {M}onte {C}arlo schemes for stochastic
  volatility models.
\newblock {\em Quantitative Finance\/}~{\em 6\/}(6), 513--536.

\bibitem[\protect\citeauthoryear{K{\"u}chler and Tappe}{K{\"u}chler and
  Tappe}{2008}]{kuchler2008bilateral}
K{\"u}chler, U. and S.~Tappe (2008).
\newblock Bilateral {G}amma distributions and processes in financial
  mathematics.
\newblock {\em Stochastic Processes and their Applications\/}~{\em 118\/}(2),
  261--283.

\bibitem[\protect\citeauthoryear{Lewis}{Lewis}{1998}]{lewis1998applications}
Lewis, A.~L. (1998).
\newblock Applications of eigenfunction expansions in continuous-time finance.
\newblock {\em Mathematical Finance\/}~{\em 8\/}(4), 349--383.

\bibitem[\protect\citeauthoryear{Li and Linetsky}{Li and
  Linetsky}{2015}]{li2015discretely}
Li, L. and V.~Linetsky (2015).
\newblock Discretely monitored first passage problems and barrier options: {A}n
  eigenfunction expansion approach.
\newblock {\em Finance and Stochastics\/}~{\em 19\/}(4), 941--977.

\bibitem[\protect\citeauthoryear{Linetsky}{Linetsky}{2004a}]{linetsky2004lookback}
Linetsky, V. (2004a).
\newblock Lookback options and diffusion hitting times: {A} spectral expansion
  approach.
\newblock {\em Finance and Stochastics\/}~{\em 8\/}(3), 373--398.

\bibitem[\protect\citeauthoryear{Linetsky}{Linetsky}{2004b}]{linetsky2004spectral}
Linetsky, V. (2004b).
\newblock The spectral decomposition of the option value.
\newblock {\em International Journal of Theoretical and Applied Finance\/}~{\em
  7\/}(03), 337--384.

\bibitem[\protect\citeauthoryear{Lions and Musiela}{Lions and
  Musiela}{2007}]{lions2007correlations}
Lions, P.-L. and M.~Musiela (2007).
\newblock Correlations and bounds for stochastic volatility models.
\newblock In {\em Annales de l'Institut Henri Poincare (C) Non Linear
  Analysis}, Volume~24, pp.\  1--16. Elsevier.

\bibitem[\protect\citeauthoryear{Lipton and Sepp}{Lipton and
  Sepp}{2008}]{lipton2008stochastic}
Lipton, A. and A.~Sepp (2008).
\newblock Stochastic volatility models and {K}elvin waves.
\newblock {\em Journal of Physics A: Mathematical and Theoretical\/}~{\em
  41\/}(34), 344012.

\bibitem[\protect\citeauthoryear{Lorig}{Lorig}{2014}]{lorig2014pricing}
Lorig, M. (2014).
\newblock Pricing derivatives on multiscale diffusions: {A}n eigenfunction
  expansion approach.
\newblock {\em Mathematical Finance\/}~{\em 24\/}(2), 331--363.

\bibitem[\protect\citeauthoryear{McNeil, Frey, and Embrechts}{McNeil
  et~al.}{2015}]{mcneil2015quantitative}
McNeil, A.~J., R.~Frey, and P.~Embrechts (2015).
\newblock {\em Quantitative Risk Management: Concepts, Techniques and Tools}.
\newblock Princeton University Press.

\bibitem[\protect\citeauthoryear{McWalter, Rudd, Kienitz, and Platen}{McWalter
  et~al.}{2017}]{mcwalter2017recursive}
McWalter, T.~A., R.~Rudd, J.~Kienitz, and E.~Platen (2017).
\newblock Recursive marginal quantization of higher-order schemes.
\newblock {\em Working Paper\/}.

\bibitem[\protect\citeauthoryear{Merton}{Merton}{1973}]{merton1973theory}
Merton, R.~C. (1973).
\newblock Theory of rational option pricing.
\newblock {\em The Bell Journal of economics and management science\/},
  141--183.

\bibitem[\protect\citeauthoryear{Mysovskikh}{Mysovskikh}{1968}]{mysovskikh1968construction}
Mysovskikh, I. (1968).
\newblock On the construction of cubature formulas with fewest nodes.
\newblock In {\em Dokl. Akad. Nauk SSSR}, Volume 178, pp.\  1252--1254.

\bibitem[\protect\citeauthoryear{Narayan, Jakeman, and Zhou}{Narayan
  et~al.}{2017}]{narayan2017christoffel}
Narayan, A., J.~Jakeman, and T.~Zhou (2017).
\newblock A {C}hristoffel function weighted least squares algorithm for
  collocation approximations.
\newblock {\em Mathematics of Computation\/}~{\em 86\/}(306), 1913--1947.

\bibitem[\protect\citeauthoryear{Nelson}{Nelson}{1990}]{nelson1990arch}
Nelson, D.~B. (1990).
\newblock {ARCH} models as diffusion approximations.
\newblock {\em Journal of Econometrics\/}~{\em 45\/}(1-2), 7--38.

\bibitem[\protect\citeauthoryear{Pag\`es and Printems}{Pag\`es and
  Printems}{2003}]{pages2003optimal}
Pag\`es, G. and J.~Printems (2003).
\newblock Optimal quadratic quantization for numerics: {T}he {G}aussian case.
\newblock {\em Monte Carlo Methods and Applications\/}~{\em 9\/}(2), 135--165.

\bibitem[\protect\citeauthoryear{Pag{\`e}s and Sagna}{Pag{\`e}s and
  Sagna}{2015}]{pages2015recursive}
Pag{\`e}s, G. and A.~Sagna (2015).
\newblock Recursive marginal quantization of the {E}uler scheme of a diffusion
  process.
\newblock {\em Applied Mathematical Finance\/}~{\em 22\/}(5), 463--498.

\bibitem[\protect\citeauthoryear{Sch{\"o}bel and Zhu}{Sch{\"o}bel and
  Zhu}{1999}]{schobel1999stochastic}
Sch{\"o}bel, R. and J.~Zhu (1999).
\newblock Stochastic volatility with an {O}rnstein--{U}hlenbeck process: {A}n
  extension.
\newblock {\em Review of Finance\/}~{\em 3\/}(1), 23--46.

\bibitem[\protect\citeauthoryear{Schoutens}{Schoutens}{2012}]{schoutens2012stochastic}
Schoutens, W. (2012).
\newblock {\em Stochastic Processes and Orthogonal Polynomials}, Volume 146.
\newblock Springer Science \& Business Media.

\bibitem[\protect\citeauthoryear{Scott}{Scott}{1987}]{scott1987option}
Scott, L.~O. (1987).
\newblock Option pricing when the variance changes randomly: {T}heory,
  estimation, and an application.
\newblock {\em Journal of Financial and Quantitative Analysis\/}~{\em 22\/}(4),
  419--438.

\bibitem[\protect\citeauthoryear{Sepp}{Sepp}{2016}]{sepp2016log}
Sepp, A. (2016).
\newblock Log-normal stochastic volatility model: {A}ffine decomposition of
  moment generating function and pricing of vanilla options.
\newblock {\em Working Paper\/}.

\bibitem[\protect\citeauthoryear{Stein and Stein}{Stein and
  Stein}{1991}]{stein1991stock}
Stein, E.~M. and J.~C. Stein (1991).
\newblock Stock price distributions with stochastic volatility: {A}n analytic
  approach.
\newblock {\em The Review of Financial Studies\/}~{\em 4\/}(4), 727--752.

\bibitem[\protect\citeauthoryear{Wiggins}{Wiggins}{1987}]{wiggins1987option}
Wiggins, J.~B. (1987).
\newblock Option values under stochastic volatility: {T}heory and empirical
  estimates.
\newblock {\em Journal of Financial Economics\/}~{\em 19\/}(2), 351--372.

\bibitem[\protect\citeauthoryear{Wilcox}{Wilcox}{1967}]{wilcox1967exponential}
Wilcox, R. (1967).
\newblock Exponential operators and parameter differentiation in quantum
  physics.
\newblock {\em Journal of Mathematical Physics\/}~{\em 8\/}(4), 962--982.

\end{thebibliography}

\end{document}